\documentclass[smallextended]{svjour3}

\usepackage{todonotes}
\usepackage{amsmath}
\usepackage{amssymb}
\usepackage{latexsym}
\usepackage{graphicx}
\usepackage{color}
\usepackage{hyperref}
\usepackage{slashbox}
\usepackage{algorithm}
\usepackage{algpseudocode}

\smartqed
\definecolor{darkblue}{rgb}{0,0,0.6}
\hypersetup{
	citecolor=darkblue,
	linkcolor=darkblue,
	urlcolor=darkblue,
	breaklinks=true,
	colorlinks=true,
	citebordercolor=0 0 0,
	filebordercolor=0 0 0,
	linkbordercolor=0 0 0,
	menubordercolor=0 0 0,
	urlbordercolor=0 0 0,
	pdfhighlight=/I,
	pdfborder=0 0 0,
	bookmarksopen=true,
	bookmarksnumbered=true
}

\newcommand{\NN}{\mathbb{N}}

\newcommand{\RR}{\mathbb{R}}

\newcommand{\twocasefunction}[4]{
	\begin{cases}
		#1 & \text{if #3} \\
		#2 & \text{if #4}
	\end{cases}
}
\newcommand{\twocaseotherwisefunction}[3]{
	\begin{cases}
		#1 & \text{if #3} \\
		#2 & \text{otherwise}
	\end{cases}
}

\newcommand{\id}{1\hspace{-1,5ex}1}

\newcommand{\BigO}[1]{\mathcal{O}(#1)}

\newcommand{\bit}{\mathfrak{b}}
\newcommand{\nsb}{\nu}
\newcommand{\nub}{\mu}
\newcommand{\bitset}{\mathcal{B}}
\newcommand{\boldone}{\mathbf{1}}
\newcommand{\boldzero}{\mathbf{0}}

\newcommand{\succedge}[3]{\mathrm{succ}_{#1}(#2,#3)}
\newcommand{\mdpval}[1]{\mathtt{val}_{#1}}
\newcommand{\mdppot}[1]{\mathtt{pot}_{#1}}
\newcommand{\pri}[1]{{\langle #1\rangle}}

\begin{document}

\markboth{D.~Avis, O.~Friedmann}{An exponential lower bound for Cunningham's rule}

\title{An exponential lower bound for Cunningham's rule}
\author{David Avis and Oliver Friedmann}
\institute{Oliver Friedmann \at University of Munich
   \\\email{me@oliverfriedmann.de}
  \and
  David Avis \at
   School of Informatics, Kyoto University, Kyoto, Japan and 
   \\ School of Computer Science, 
    McGill University, Montr{\'e}al, Qu{\'e}bec, Canada \\\email{avis@cs.mcgill.ca}}

\maketitle

\begin{abstract}
In this paper we give an exponential lower bound for Cunningham's least recently considered (round-robin)
rule as applied to parity games, Markhov decision processes and linear programs. 
This improves a recent subexponential bound 
of Friedmann for this rule on these problems.
The round-robin rule fixes a cyclical order of the variables and chooses the next pivot variable
starting from the previously chosen variable and proceeding in the given circular order. It is perhaps the
simplest example from the class of history-based pivot rules.         

Our results are based on a new lower bound construction for parity games.
Due to the nature of the construction we are also able to obtain an exponential lower bound
for the round-robin rule applied to acyclic unique sink
orientations of hypercubes (AUSOs). Furthermore these AUSOs are realizable as polytopes. 
We believe these are the first such results for history based rules for AUSOs, realizable or not.

The paper is self-contained and requires no previous knowledge of parity games.

\end{abstract}







\maketitle

\section{Introduction}

The search for a polynomial time pivoting rule for the simplex method is as old as the method itself.
Klee and Minty showed in 1970 \cite{Klee1972} that Dantzig's original rule was exponential and similar results 
were soon found for most of the other known rules. All such lower bound
constructions were based on variations of the deformed hypercube 
that appeared in Klee and Minty's original paper. The constructions have the property that
some variables pivoted only a very few times - sometimes only once - in the exponential pivot path. 
This motivated Cunningham \cite{Cunningham79}, Zadeh \cite{Zadeh80} and
others to consider so-called history based rules. 
See \cite{AADMM12} for a formal description of these and several other history based rules.

Cunningham's {\it least recently consider rule (round-robin rule)}
assigns a cyclic order to the variables and remembers the last variable to enter the basis.
The next entering variable is chosen to be the first allowable candidate 
starting from the last chosen variable and following the given
circular order. Zadeh's {\it least entered rule} chooses the entering variable to be the candidate that
has entered the basis least often. History based rules defeat the deformed hypercube
constructions because they tend to average out how many times a variable pivots.
This pseudo-random behaviour held out the possibility that they might be at worst subexponential,
if not polynomial, since the random facet rule \cite{Kalai92} 
\cite{MSW92} is subexponential.

Friedmann gave the first evidence that history-based rules can sometimes be non-polynomial
by showing the least entered rule \cite{Friedmann11} and the
round-robin rule \cite{Friedmann12} are superpolynomial in the worst case.
These results were obtained by first
constructing a certain two person game, known as a parity game, for which the players (zero and one) follow a superpolynomial number of moves. These games are then related to
Markhov decision processes (MDPs) and finally to linear programs (LPs).
It is shown that each strategy in the parity game corresponds to a vertex in
the derived LP and improving from one strategy to the next corresponds to a pivot step in the LP.

In this paper we first obtain a new lower bound construction for parity games. We define
a strategy improvement rule for player zero that corresponds to the least recently considered rule and show
an exponential lower bound on the number of 
strategy improvements made to complete the game. Using the earlier transformations this
shows an exponential lower bound for MDPs and LPs using this rule. 

However the nature of the new construction allows us
to do more. An acyclic unique sink orientation of a hypercube (AUSO) \cite{SzWe01} is an orientation of 
the hypercube's edges
so that the resulting directed graph has no cycles and each face of the hypercube has a unique sink 
(vertex of outdegree
zero). AUSOs appear in many applications and are an abstraction of linear programming itself \cite{GaSc06}.
LP pivot rules have natural analogues on AUSOs and their analysis has been the subject of 
several papers. For example in \cite{MatousekS04} an exponential lower bound is given
for the random edge pivot selection rule. We are able to show an exponential lower bound for the
least recently considered rule on AUSOs, which are realizable as LPs.

The paper is organized as follows. In the next section we
begin by defining parity games and policy iteration, and
present the new lower bound construction. The longest part of the paper is a proof that this
game requires an exponential number of moves before terminating. The section is concluded with
some applications of this result to other types of games. 
In Section \ref{mdp} we review the known connection between parity games and Markhov decision processes.
This gives rise to an exponential lower bound on MDPs that use an analogue of Cunningham's rule.
In Section \ref{lp} we again exploit a known connection to obtain explicit LPs from the MDP examples.
The least recently considered pivot rule on these LPs is shown to require an exponential number
of steps. We turn to AUSOs in Section \ref{auso}. Here we show that the parity game exhibited in
Section \ref{parity} gives a natural acyclic orientation of a hypercube built on player zero's 
strategies. Cunningham's rule on this AUSO follows an exponentially long path. Furthermore
we show that the each AUSO can be realized as a polytope. In fact the polytope is
the one that arises from the same parity game after transforming it to a MDP and then to an LP.
The paper concludes with some open problems for future research.

\section{Parity Game Policy Iteration Lower Bound}
\label{parity}

This section is organized as follows. We first define parity games and how the general 
strategy improvement algorithms 
operate on them. Since some readers may not be familiar with this material we will illustrate them on
an example that will later be generalized for the lower bound results.
We then describe the lower bound construction and prove it correct. 
Finally, we show how to extend the results to related game classes.

\subsection{Parity Games}

A \emph{parity game} is a tuple $G = (V,V_0,V_1,E,\Omega)$ where $(V,E)$ forms a directed graph whose node set is partitioned into $V = V_0 \cup V_1$ with $V_0 \cap V_1 = \emptyset$, and $\Omega : V \to \NN$ is the \emph{priority function} that assigns to each node a natural number called the \emph{priority} of the node. We assume the graph to be total, i.e.\ for every $v \in V$ there is a $w \in V$ s.t.\ $(v,w) \in E$.

We depict parity games as directed graphs where nodes owned by player 0 are drawn as circles and nodes owned by player 1 are drawn as rectangles; all nodes are labelled with their respective name and priority. 
An example of such a graph is shown in
Figure~\ref{figure: example}.
For each node the name, $a_2, F_1, t, \cdots $, is on top and the priority is underneath.
For the moment we ignore the colours on the edges and that some edges are shown dashed.
(For monochromatic figures, we ignore that some edges are bold and some are dashed.)
A very important property of this game is that the out-degree of each node belonging
to player 0 is two. We call a game with this property a \emph{binary game}.

\begin{figure}[h!]
\centering
\includegraphics[scale=0.58]{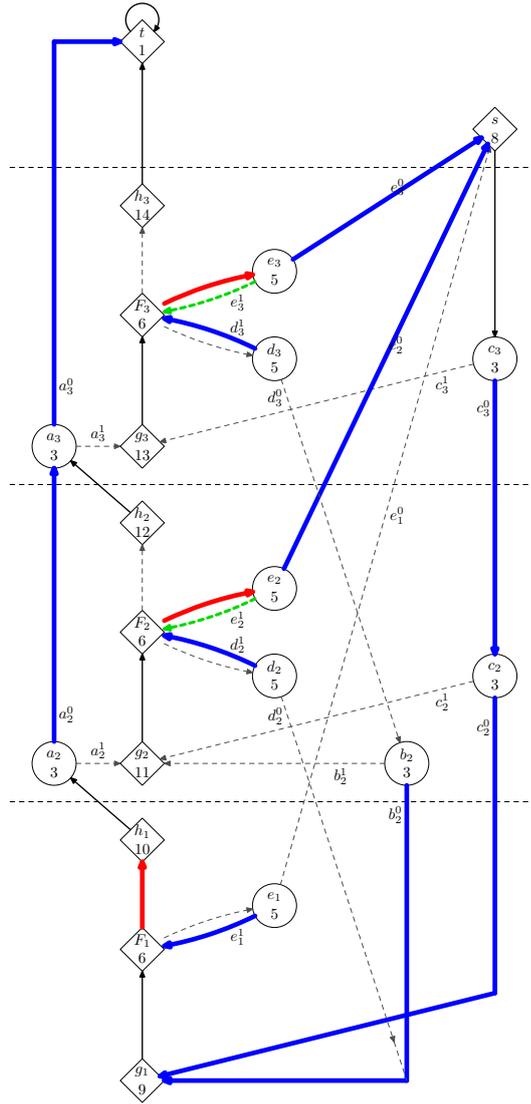}
\caption{Parity Game Lower Bound Graph $G_3$}
\label{figure: example}
\end{figure}
%


We use infix notation $vEw$ instead of $(v,w) \in E$ and define the set of all \emph{successors} of $v$ as $vE := \{ w \mid vEw \}$. The size $|G|$ of a parity game $G = (V,\ V_0,\ V_1,\ E,\ \Omega)$ is defined to be the cardinality of $E$, i.e.\ $|G| := |E|$; since we assume parity games to be total w.r.t.\ $E$, this is a reasonable way to measure the size.
The example has size 36 and $F_2 E := \{ h_2 , e_2 , d_2 \}$. 

The game is played between two players called $0$ and $1$: starting at a node $v_0 \in V$, they construct an 
infinite path through the graph as follows. If the construction so far has yielded a finite sequence $v_0\ldots v_n$ and $v_n \in V_i$ then player $i$ selects a $w \in v_nE$ and the play continues with $v_0\ldots v_n w$.
In the example a game may have started as the sequence $a_2 g_2 F_2 d_2$ ending at a node owned by player
0. She can choose between $F_2$ and $g_1$ and could continue by appending either node to the sequence.

Every play has a unique winner given by the \emph{parity} of the greatest priority that occurs infinitely often. The winner of the play $v_0 v_1 v_2 \ldots$ is player $i$ iff $\max \{ p \mid \forall j \in \NN\, \exists k \geq j:\, \Omega(v_k) = p \} \equiv_2 i$\footnote{$x \equiv_2 y$ if and only if $x$ and $y$ are congruent mod 2}.
That is, player 0 tries to make an even priority occur infinitely often without any greater odd priorities occurring infinitely often, player 1 attempts the converse.
In the example we may consider the infinite path $a_2 g_2 F_2 d_2 F_2 d_2 ...$. In this case 6 is
the largest priority that occurs infinitely often and player 0 wins since this number is even.

A \emph{strategy} for player $i$ is a -- possibly partial -- function $\sigma: V^*V_i \to V$, s.t.\ for all sequences $v_0 \ldots v_n$ with $v_{j+1} \in v_jE$ for all $j=0,\ldots,n-1$, and all $v_n \in V_i$ we have: $\sigma(v_0\ldots v_n) \in v_nE$. A play $v_0 v_1 \ldots$ \emph{conforms} to a strategy $\sigma$ for player $i$ if for all $j \in \NN$ we have: if $v_j \in V_i$ then $v_{j+1} = \sigma(v_0\ldots v_j)$. Intuitively, conforming to a strategy means to always make those choices that are prescribed by the strategy. A strategy $\sigma$ for player $i$ is a \emph{winning strategy} in node $v$ if player $i$ wins every play that begins in $v$ and conforms to $\sigma$.

A strategy $\sigma$ for player $i$ is called \emph{positional} if for all $v_0\ldots v_n \in V^*V_i$ and all $w_0\ldots w_m \in V^*V_i$ we have: if $v_n = w_m$ then $\sigma(v_0\ldots v_n) = \sigma(w_0\ldots w_m)$. That is, the choice of the strategy on a finite path only depends on the last node on that path.
So in this case we need only specify $\sigma(v)$ for each $v \in V$.
The set of positional strategies for player $i$ is denoted by $\mathcal{S}_i(G)$.
In the example a partial positional strategy could consist of $\sigma(a_2)=g_2$,    
$\sigma(g_2)=F_2$, $\sigma(F_2)=d_2$, $\sigma(d_2)=F_2$ and player 0 wins.
Note that we do not need to give a strategy for out-degree one nodes, such as $g_2$, and
will omit these in the sequel.

Recall that a \emph{binary game} is one where each node belonging to player 0 has out-degree two.
In this case player 0's positional strategy has a very 
simple representation. Suppose she owns $n$ nodes and labels
the out-edges for each of them 0 or 1 in any arbitrary way. Then her positional strategy can be represented
as a binary $n$-vector specifying for each node which edge is chosen in the strategy.

With $G$ we associate two sets $W_0,W_1 \subseteq V$, where $W_i$ is the set of all nodes $v$ where player
$i$ wins the game $G$ starting at $v$. Here we may restrict ourselves to positional strategies because it is
well-known that a player has a (general) winning strategy if and only if she has a positional winning strategy for a
given game. In fact, parity games enjoy positional determinacy meaning that for every node
$v$ in the game either $v \in W_0$ or $v \in W_1$ \cite{focs91*368}. Furthermore,
it is not difficult to show that, whenever player $i$ has winning strategies $\sigma_v$ for all $v \in U$ for
some $U \subseteq V$, then there is also a single strategy $\sigma$ that is winning for player $i$ from every
node in $U$.

The problem of solving a parity game is to compute $W_0$ and $W_1$ as well as corresponding winning
strategies $\sigma_0$ and $\sigma_1$ for the players on their respective winning regions.
In the example $F_3 \in W_1$ since player 1 can set $\sigma(F_1) = h_3$ and end up in the infinite loop
on $t$ with priority one which is odd. 


\subsection{Strategy Improvement}

We describe here the basic definitions of the strategy improvement algorithm.
For a given parity game $G = (V,\ V_0,\ V_1,\ E,\ \Omega)$, the \emph{reward} of node $v$ is defined as follows: $rew(v) := \Omega(v)$ if $\Omega(v) \equiv_2 0$ and $rew(v) := -\Omega(v)$ otherwise. The set of \emph{profitable nodes} for player 0 resp.\ 1 is defined to be $V_\oplus := \{v \in V \mid \Omega(v) \equiv_2 0\}$ resp.\ $V_\ominus := \{v \in V \mid \Omega(v) \equiv_2 1\}$.

The \emph{relevance ordering} $<$ on $V$ is induced by $\Omega$: $v < u :\iff \Omega(v) < \Omega(u)$.
Additionally one defines the \emph{reward ordering} $\prec$ on $V$ by $v \prec u :\iff rew(v) < rew(u)$. 
In our construction, although priorities are not unique, they are unique on each cycle.
Therefore on each cycle both orderings are total.

Let $v$ be a node, $\sigma$ be a positional player~0 strategy and $\tau$ be a positional player~1 strategy. Starting at $v$, there is exactly one path $\pi_{\sigma,\tau,v}$ that conforms to $\sigma$ and $\tau$. Since $\sigma$ and $\tau$ are positional strategies, this path can be uniquely written as follows.
\begin{displaymath}
\pi_{\sigma,\tau,v} = v_1\ldots v_k(w_1\ldots w_l)^\omega
\end{displaymath}
The superscript $\omega$ denotes an infinite cycle on the given vertex or vertices.
Here $v_1 \ldots v_k$ is a (possibly empty) non-repeating set of vertices and $w_1\ldots w_l$ is an infinite
cycle with  $w_1 > w_j$ for all $1 < j \leq l$. If $k \ge 1$ then $v = v_1$ otherwise $v$ is a member of 
the cycle.
Note that the uniqueness follows from the fact that all nodes on the cycle have different priorities.

Discrete strategy improvement relies on a more abstract description of such a play $\pi_{\sigma,\tau,v}$. In fact, we only consider the \emph{dominating cycle node} $w_1$, the set of \emph{more relevant nodes} -- i.e.\ all $v_i > w_1$ -- on the path to the cycle node, and the \emph{length} $k$ of the path leading to the cycle node. 

The \emph{node valuation of $v$ w.r.t.\ $\sigma$ and $\tau$} is defined as follows.
\begin{displaymath}
\vartheta_{\sigma,\tau,v} := (w_1, \{v_i > w_1 \mid 1 \leq i \leq k\}, k)
\end{displaymath}
Given a node valuation $\vartheta$, we refer to $w_1$ as the \emph{cycle component}, to $\{v_i > w_1 \mid 1 \leq i \leq k\}$ as the \emph{path component}, and to $k$ as the \emph{length component} of $\vartheta$.

In the example if we have $\sigma(a_3)=g_3$, $\tau(F_3)=d_3$ and $\sigma(d_3)=F_3$ then
$\pi_{\sigma,\tau,h_2} = h_2 a_3 g_3 (F_3 d_3)^\omega$ and 
node valuation $\vartheta_{\sigma,\tau,h_2} := (F_3, \{h_2, g_3\}, 3)$.
$\{F_3, d_3 \}$ is the cycle component with dominating node $F_3$, $\{h_2,a_3,g_3\}$ is the path to the cycle
and has length 3, $\{h_2,g_3\}$ are the more relevant nodes.

In order to compare node valuations with each other, we introduce a partial ordering 
$\prec$ on the set of node valuations. 
For this we first define a partial ordering on the path components.
For a set of nodes $M$, let $\Omega(M)$ denote the \emph{priority occurence mapping}:
\[
	\Omega(M): p \mapsto |\{v \in M \mid \Omega(v) = p \}|
\]
Let $M$ and $N$ be two distinct sets of vertices. If $\Omega(M) = \Omega(N)$ then
$M$ and $N$ are not comparable and we write this $M \sim N$.
Otherwise let $p$ be the highest priority s.t.\ $\Omega(M)(p) \not= \Omega(N)(p)$.
Then $M \prec N$ if $\Omega(M)(p) > \Omega(N)(p)$ and $p\equiv_2 1$, or $\Omega(M)(p) < \Omega(N)(p)$ and $p\equiv_2 0$.  
Otherwise $N \prec M$.
We observe that in a game with all priorities unique $\prec$ defines a total order on
the subsets of vertices.
 

In the example, if we compare 
$M=\{c_3, c_2, g_2\}$ with 
$N=\{h_1, a_2, g_2\}$, we have 
$\Omega(M)=\{3\mapsto 2,11\mapsto 1\}$ and
$\Omega(N)=\{10\mapsto 1,3\mapsto 1,11\mapsto 1\}$.
The priority with maximum value in which both sets differ is 10.
Since this priority is even we have
$M \prec N$.

Now we are able to extend the partial ordering on sets of nodes to node valuations.
If $u \prec v$ then $(u, M, e) \prec (v, N, f)$.
Otherwise $rew(u) = rew(v)$ and
\begin{displaymath}
(u, M, e) \prec (v, N, f)  \iff
\begin{cases}
 M \prec N   \\
 M \sim N \textrm ,~ e < f \textrm{ and } u \in V_\ominus  \\
 M \sim N \textrm ,~ e > f \textrm{ and } u \in V_\oplus 
\end{cases}
\end{displaymath}

We write $(u, M, e) \sim (v, N, f)$ iff neither $(u, M, e) \prec (v, N, f)$ nor $(v, N, f) \prec (u, M, e)$.
We write $(u, M, e) \preceq (v, N, f)$ to abbreviate $(u, M, e) \prec (v, N, f)$ or $(u, M, e) \sim (v, N, f)$.  

We observe that if all priorities are unique then we have a total order on node valuations.
For in this case if $rew(u)=rew(v)$ then $u=v$ and if $M \sim N$ then $M=N$.
We cannot have $e=f$ for otherwise $(u, M, e) = (v, N, f)$.

The motivation behind the above ordering is a lexicographic measurement of the profitability of a positional play w.r.t.\ player~0: the most prominent part of a positional play is the cycle in which the plays eventually stays, and here it is the reward ordering on the dominating cycle node that defines the profitability for player~0. The second important part is the loopless path that leads to the dominating cycle node. Here, we measure the profitability of a loopless path by a \emph{lexicographic} ordering on the \emph{relevancy} of the nodes on path, applying the \emph{reward} ordering on each component in the lexicographic ordering. Finally, we consider the length, and the intuition behind the definition is that, assuming we have an even-priority dominating cycle node, it is better to reach the cycle fast whereas it is better to stay as long as possible out of the cycle otherwise.

In the example suppose $\sigma$ and $\tau$ give rise to the paths $c_3 c_2 g_2 (F_2 d_2)^{\omega}$
and $h_1 a_2 g_2 (F_2 d_2)^{\omega}$. We have node evaluations
$\vartheta_{\sigma,\tau,h_1} := (F_2, \{h_1, g_2\}, 3)$, 
$\vartheta_{\sigma,\tau,c_3} := (F_2, \{g_2\}, 3)$ and
$\vartheta_{\sigma,\tau,c_2} := (F_2, \{g_2\}, 2)$.
We have
\[
 (F_2, \{g_2\}, 2) \prec (F_2, \{g_2\}, 3) \prec  (F_2, \{h_1, g_2\}, 3)
\]
The first $\prec$ is due to the fact that the length of the path to the cycle node
is smaller in the first node evaluation.
The second $\prec$ is because the symmetric difference of the path components is $h_1$ 
with priority 10 which is even.

Given a player~0 strategy $\sigma$, it is player 1's goal
to find a best response counter-strategy $\tau$ that minimizes the associated node valuations.
A strategy $\tau$ is an \emph{optimal counter-strategy} w.r.t.\ $\sigma$ iff for every opponent strategy $\tau'$ and for every node $v$ we have: $\vartheta_{\sigma,\tau,v} \preceq \vartheta_{\sigma,\tau',v}$.

Is is well-known that an optimal counter-strategy always exists and that it is efficiently computable.
\begin{lemma}[\cite{conf/cav/VogeJ00}] \label{update in poly time}
Let $G$ be a parity game and $\sigma$ be a player~0 strategy. An optimal counter-strategy for player~1 w.r.t.\ $\sigma$ exists and can be computed in polynomial time.
\end{lemma}

A fixed but arbitrary optimal counter-strategy will be denoted by $\tau_\sigma$ from now on. The associated \emph{game valuation} $\Xi_\sigma$ is a map that assigns to each node the node valuation w.r.t.\ $\sigma$ and $\tau_\sigma$:
\begin{displaymath}
\Xi_\sigma: v \mapsto \vartheta_{\sigma,\tau_\sigma,v}
\end{displaymath}

In the example suppose player 0 has played to obtain the path $a_3 g_3 (F_3 d_3)^{\omega}$.
It is easy to see that $\tau_{\sigma}(F_3)=h_3$ is an optimal
counter-strategy and $\vartheta_{\sigma,\tau_{\sigma},a_3} = (t,\{a_3,g_3,F_3,h_3\},4)$.
We have $\Xi_\sigma (a_3) = (t,\{a_3,g_3,F_3,h_3\},4)$ and player 1 wins.

Game valuations are used to measure the performance of a strategy of player~0. 
For a fixed strategy $\sigma$ of player~0 and a node $v$, the associated valuation essentially states which is the worst cycle that can be reached from $v$ conforming to $\sigma$ as well as the worst loopless path leading to that cycle (also conforming to $\sigma$).
We also write $v \prec_\sigma u$ to compare the $\Xi_\sigma$-valuations of two nodes, i.e.\ to abbreviate $\Xi_\sigma(v) \prec \Xi_\sigma(u)$.

A run of the strategy improvement algorithm can be expressed by a sequence of \emph{improving} game valuations; a partial ordering on game valuations is quite naturally defined as follows:
\begin{displaymath}
\Xi \lhd \Xi' \, : \iff \, \left(\Xi(v) \preceq \Xi'(v) \textrm{ for all } v \in V\right) \textrm{ and } \left(\Xi \not= \Xi'\right)
\end{displaymath}

Let $\sigma$ be a strategy, $v \in V_0$ and $w \in vE$. We say that $(v,w)$ is
a \emph{$\sigma$-improving switch} iff $\sigma(v) \prec_\sigma w$.
If $w \prec_\sigma \sigma(v)$, we say that $(v,w)$ is a \emph{$\sigma$-degradable switch}.
We say that
$\sigma$ is \emph{improvable} iff $\sigma$ has an improving switch.
We write $I_\sigma$ to denote the set of improving switches and write
$I_\sigma(v) = \{w \mid (v,w) \in I_\sigma\}$.

Again things become very simple for binary games. For any node $v$ we can 
write $vE=\{ \sigma(v), w \}$
and an improving switch selects edge $(v,w)$
if $\sigma(v) \prec_\sigma w$.
Clearly the notion of improving switch requires that the node valuations of $\sigma(v)$
and $w$ are ordered by $\prec_\sigma$.
The binary games we construct do not have unique priorities and $\prec$ defines only a partial order
on the node valuations. 

Recall that for binary games player 0's current strategy can be represented by a binary $n$-vector.
As improving switch just flips one of the bits in this vector. The connection with
paths on hypercubes now becomes apparent.

In the example consider strategies of the two players leading to the infinite path
$d_2 g_1 F_1 h_1 a_2 a_3 t^{\omega}$ with node valuation
\[
\vartheta_{\sigma,\tau,d_2} := (t, \{d_2, g_1, F_1, h_1, a_2, a_3\}, 6).
\]
Then $(d_2, F_2)$ is an improving switch for player 0 since it leads to the infinite path $(d_2 F_2)^{\omega} $
with node valuation 
$\vartheta_{\sigma,\tau,d_2} := (F_2, \emptyset , 0)$ which is better for her since $t \prec F_2$.

The improvement step from one strategy to the next is carried out by an
\emph{improvement rule}. It is a map $\mathcal{I}_G: \mathcal{S}_0(G) \rightarrow \mathcal{S}_0(G)$
s.t.\ $\Xi_\sigma \lhd \Xi_{\mathcal{I}_G(\sigma)}$ for every $\sigma$ and additionally
 $\Xi_\sigma \unlhd \Xi_{\mathcal{I}_G(\sigma)}$ if $\sigma$ is improvable.
We say that a function $\mathcal{I}_G: \mathcal{S}_0(G) \rightarrow \mathcal{S}_0(G)$
is a \emph{standard improvement rule} iff it only selects improving switches
for finding a successor strategy, i.e.\
\begin{enumerate}
\item For every node $v \in V_0$ it holds that $\sigma(v) \preceq_\sigma \mathcal{I}_G(\sigma)(v)$.
\item If $\sigma$ is improvable then there is a node $v \in V_0$ s.t.\ $\sigma(v) \prec_\sigma \mathcal{I}_G(\sigma)(v)$.
\end{enumerate}

Jurdzi{\'n}ski and V\"oge \cite{conf/cav/VogeJ00} showed that 
an improving switch always exists for any non-optimal strategy $\sigma$.
They also showed that improving $\sigma$ by
an arbitrary, non-empty selection of improving switches can only result in strategies with valuations strictly better than the valuation of $\sigma$.

\begin{theorem}[\cite{conf/cav/VogeJ00}]
Let $G$ be a parity game, $\sigma$ be an improvable strategy and $\mathcal{I}_G$ be a standard improvement rule. Then $\Xi_\sigma \lhd \Xi_{\mathcal{I}_G(\sigma)}$.
\end{theorem}

If a strategy is not improvable, the strategy improvement procedure comes to an end.
The game has been solved.
The winning sets for both players as well as associated winning strategies can be easily derived from the given valuation.

\begin{theorem}[\cite{conf/cav/VogeJ00}] \label{winning theorem}
Let $G$ be a parity game and $\sigma$ be a non-improvable strategy. Then the following holds:
\begin{enumerate}
\item $W_0 = \{v \mid \Xi_\sigma(v) = (w, \_, \_) \textrm{ and } w \in V_\oplus\}$
\item $W_1 = \{v \mid \Xi_\sigma(v) = (w, \_, \_) \textrm{ and } w \in V_\ominus\}$
\item $\sigma$ is a winning strategy for player 0 on $W_0$
\item $\tau_\sigma$ is a winning strategy for player 1 on $W_1$
\item $\sigma$ is $\unlhd$-optimal
\end{enumerate}
\end{theorem}

Strategy improvement starts with an initial strategy $\sigma$ and runs for a
given improvement rule $\mathcal{I}$ as follows and returns an optimal
player~0 strategy as outlined in the pseudo-code of
Algorithm~\ref{algorithm: discrete strategy improvement}.

\begin{algorithm}
\begin{algorithmic}[1]
\Procedure{StandardStratIt}{$\mathcal{I}$, $G$, $\sigma$}
	\While {$\sigma$ is improvable}
		\State $\sigma \gets \mathcal{I}_G(\sigma)$
	\EndWhile
	\State \textbf{return} $\sigma$.
\EndProcedure	
\end{algorithmic}
\caption{Strategy Improvement}
\label{algorithm: discrete strategy improvement}
\end{algorithm}

Given an initial strategy $\sigma$, a game $G$ and a rule $\mathcal{I}$, the
unique execution trace, called \emph{run}, of strategy improvement is the sequence of strategies
$\sigma_1$, $\ldots$, $\sigma_k$ s.t.\ $\sigma_1 = \sigma$, 
$\sigma_{i+1} = \mathcal{I}_G(\sigma_i)$ for all $i < k$, $\sigma_k$
optimal and $\sigma_i$ improvable for all $i < k$. The \emph{length} of the run
is denoted by $k$ and we say that strategy improvement \emph{requires $k+1$ iterations to find the optimal strategy}.

We call a parity game $G$ 
(in combination with an initial strategy $\theta$) 
a \emph{sink game} iff the following two properties hold:
\begin{enumerate}
\item \emph{Sink Existence}: there is a node $v^*$ (called the \emph{sink} of $G$) with $v^*Ev^*$ and $\Omega(v^*) = 1$ reachable from all nodes; also, there is no other node $w$ with $\Omega(w) \leq \Omega(v^*)$.
\item \emph{Sink Seeking}: for each player~0 strategy $\sigma$ with $\Xi_{\theta} \unlhd \Xi_\sigma$ and each node $w$ it holds that the cycle component of $\Xi_\sigma(w)$ equals $v^*$.
\end{enumerate}

At this point the reader may wish to verify that the example is a sink game with $v^* = t$.
Obviously, a sink game is won by player~1. Note that comparing node valuations in a sink game can be reduced to comparing the path components of the respective node valuations, for two reasons. First, the cycle component remains constant. Second, the path-length component equals the cardinality of the path component, because all nodes except the sink node are more relevant than the cycle node itself.
In the case of a sink game, we will therefore identify node valuations with their path component.

Given a parity game $G$ 
the sink existence property can be verified by standard graph algorithms.
Given an initial strategy $\theta$
the sink seeking property can also be easily checked, as shown by the following lemma.

\begin{lemma}[\cite{FriedmannExpBound09}] \label{lemma: sink game lemma} Let $G$ be a parity game 
with initial strategy $\theta$ fulfilling the sink existence property w.r.t.\ $v^*$. $G$ is a sink game iff $G$ is completely won by player~1 (i.e.\ $W_1 = V$) and for each node $w$ it holds that the cycle component of $\Xi_{\theta}(w)$ equals $v^*$.
\end{lemma}

Let $G$ be a sink game and $v, r \in V_G$. We define $\Xi_\sigma^{>r}(v)$ to be
the path component of $\Xi_\sigma(v)$ by filtering the nodes which are more relevant
than $r$, i.e.\
\begin{displaymath}
	\Xi_\sigma^{>r}(v) = \{u \in \Xi_\sigma^{>r}(v) \mid \Omega(u) > \Omega(r)\}
\end{displaymath}
It is easy to see that $\Xi_\sigma^{>r}(v) \prec \Xi_\sigma^{>r}(u)$ implies
$\Xi_\sigma(v) \prec \Xi_\sigma(u)$.
We assume from now on that every game we consider is a sink game.

Cunningham's rule \cite{Cunningham79} is a \emph{deterministic} \emph{history based}
pivot rule for selecting entering variables in the network simplex method. 
It fixes an
initial ordering on all variables and then selects the entering variables
in a round-robin fashion starting from the last entering variable selected.
The history is simply to remember this variable.
The rule can be adapted in a straightforward manner to other local improvement algorithms.

We describe Cunningham's pivoting rule in the context of parity games. We
assume that we are given a total ordering $\prec$ on the player~$0$ edges of the
parity game. The history is simply to record the last edge that has been applied.

Given a non-empty subset of player~$0$ edges $\emptyset \not= F \subseteq E_0$
and a player~$0$ edge $e \in E_0$, we define a \emph{successor operator} as follows:
\begin{displaymath}
	\succedge{\prec}{e}{F} := \begin{cases}
		\min_\prec \{e' \in F \mid e \preceq e'\} & \text{ if } \{e' \in F \mid e \preceq e'\} \not= \emptyset \\
		\min_\prec \{e' \in F \mid e' \preceq e\} & \text{ otherwise }
	\end{cases}
\end{displaymath}

See Algorithm~\ref{algorithm: cunningham} for a pseudo-code specification
of the Cunningham's rule for applied to parity games. 

\begin{algorithm}[!h]
\begin{algorithmic}[1]
\Procedure{RoundRobin}{$G$, $\sigma$, $\prec$, $e$}
	\While {$\sigma$ is improvable}
		\State $e \gets \succedge{\prec}{e}{I_\sigma}$
		\State $\sigma \gets \sigma[e]$
	\EndWhile
	\State \textbf{return} $\sigma$.
\EndProcedure
\end{algorithmic}
\caption{Cunningham's Improvement Algorithm}
\label{algorithm: cunningham}
\end{algorithm}

Let $(\sigma_1,e_1)$, $\ldots$, $(\sigma_n,e_n)$ be a trace of the algorithm w.r.t.\
some selection ordering $\prec$. We write $(\sigma,e) \leadsto_\prec (\sigma',e')$
iff there are $i < j$ s.t.\ $(\sigma,e) = (\sigma_i,e_i)$ and $(\sigma',e') = (\sigma_j,e_j)$.

In the original specification of Cunningham's rule \cite{Cunningham79} it is
assumed that the ordering on the edges and the
initial edge $e$ is given as part of the input.
In fact, we know that the asymptotic behavior of Cunningham's improvement rule
highly depends on the ordering used, at least in the world
of parity games and strategy improvement for games in general. We have the following
theorem which is easy to verify (the idea is that there is at least one
improving switch towards the optimal strategy in each step).

\begin{theorem}\label{theorem: linear diameter}
Let $G$ be a parity game with $n$ nodes and $\sigma_0$ be a strategy. There is a 
sequence policies $\sigma_0,\sigma_1,\ldots,\sigma_{N}$ and a sequence of \emph{different} switches
$e_1,e_2,\ldots,e_N$ with $N\leq n$ s.t.\ $\sigma_{N-1}$ is optimal, $\sigma_{i+1} = \sigma_i[e_{i+1}]$
and $e_{i+1}$ is an $\sigma_i$-improving switch.
\end{theorem}

Since all switches are different in the sequence, it follows immediately that
there is always a way to select an ordering that results in a linear number of pivoting
steps to solve a parity game with Cunningham's improvement rule. However, there is no
obvious method to efficiently find such an ordering. 
In order to derive a lower bound we are entitled to give both the input graph and the ordering
to be used.

\subsection{Lower Bound Construction}

Our lower bound construction is a natural generalization of the parity game $G_3$, shown in 
Figure~\ref{figure: example}, that we used throughout the last subsection.
For each $n \ge 3$ we define the underlying graph $G_n = (V_0,V_1,E,\Omega)$ 
as follows.
\begin{flalign*}
V_0 \;:=\;& \{a_i, c_i, d_i \mid 1 < i \leq n\} \cup \{b_i \mid 1 < i < n\} \cup \{e_i \mid 1 \leq i \leq n\} \\
V_1 \;:=\;& \{F_i \mid 1 \leq i \leq n\} \cup \{g_i,h_i \mid 1 \leq i \leq n\} \cup \{s,t\}
\end{flalign*}
Figure~\ref{table: parity game lower bound edges} defines the edge sets and the priorities of $G_n$.
For convenience of notation, we identify the node names
$a_{n+1}$ with $t$,
$b_1$ with $g_1$, and
$c_1$ with $g_1$.
Explicit constructions of $G_n$ for small $n$ are available online \cite{online}.

\begin{figure}[h]
\centering
\begin{tabular}[ht]{|c||c|c||c||c|c|}
  \hline
  Node & Successors & Priority & Node & Successors & Priority \\
  \hline
  \hline
  $a_i$ & $g_i$, $a_{i+1}$    & $3$		&		$F_i$ & $h_i$, $d_{i>1}$, $e_i$ & $6$ \\
  $b_i$ & $g_i$, $b_{i-1}$    & $3$ 	&		$g_i$ & $F_i$               & $2 \cdot i + 7$ \\
  $c_i$ & $g_i$, $c_{i-1}$    & $3$		&		$h_i$ & $a_{i+1}$           & $2 \cdot i + 8$ \\
  $d_i$ & $F_i$, $b_{i-1}$    & $5$		&		$s$   & $c_n$               & $8$ \\
  $e_i$ & $F_i$, $s$          & $5$		&		$t$   & $t$                 & $1$ \\
  \hline
\end{tabular}
\caption{{Parity Game Lower Bound Graph}}
\label{table: parity game lower bound edges}
\end{figure}

\begin{lemma}
For every $n$, the game $G_n$ is a binary sink parity game.
\end{lemma}

To avoid special cases we assume $n \ge 3$. It is easy to verify that
the total number of nodes is $G_n$ is
$8n - 3$, the total number of edges is $15n - 9$, the number of different priorities is 
$2n + 5$ and the highest priority is $2n + 8$. Therefore we have $|G_n| \in \BigO{n}$.

We refer to the edges of player~0 by the names given in Figure~\ref{table: parity game lower bound edge names}.
For convenience of notation, we write $\sigma(a_i) = j$ to indicate that $a_i^j \in \sigma$ etc.

\begin{figure}[h]
\centering
\begin{tabular}[ht]{|c||c|c||c||c|c|}
  \hline
  Name & Node & Successor & Name & Node & Successor \\
  \hline
  \hline
  $a_i^1$ & $a_i$ & $g_i$ & $a_i^0$ & $a_i$ & $a_{i+1}$ \\
  $b_i^1$ & $b_i$ & $g_i$ & $b_i^0$ & $b_i$ & $b_{i-1}$ \\
  $c_i^1$ & $c_i$ & $g_i$ & $c_i^0$ & $c_i$ & $c_{i-1}$ \\
  $d_i^1$ & $d_i$ & $F_i$ & $d_i^0$ & $d_i$ & $b_{i-1}$ \\
  $e_i^1$ & $e_i$ & $F_i$ & $e_i^0$ & $e_i$ & $s$ \\
  \hline
\end{tabular}
\caption{{Parity Game Lower Bound Player~0 Edges}}
\label{table: parity game lower bound edge names}
\end{figure}

From a high-level point of view, a run of the strategy improvement algorithm mimics the counting process 
of a binary counter, yielding an exponential number of steps. Obviously, the specifics of the run depend on our choice of the edge-ordering. Every strategy that we obtain during a run corresponds to the state of the binary counter. However, a single increment step of the binary counter corresponds to several consecutive improvements in strategy improvement. 
These intermediate steps can be partitioned into well-defined \emph{phases}.

Before describing the phases in detail, consider the layout of the game graph. It is separated into uniform layers that correspond to the different bits of the binary counter.
The first and the last bit have less nodes than all the other bits.
We could include the additional nodes in the game graph but they would be of no use for the counting process. Every (disregarding the first and the last) such layer contains five player~0 nodes $a_i$, $b_i$, $c_i$, $d_i$, $e_i$, and
three player~1 nodes $F_i$, $g_i$, $h_i$.

The general construction extends Figure \ref{figure: example} in a natural way.
We use terms such as up, down, left, right, etc. based on the layout used in the figure.
The $a_i$-nodes build up a ladder-like structure that connect layers with each other, starting from the 
least-significant to the most-significant bit. For every layer, player~0 has to the choice to either enter the layer or to 
directly pass on to the next layer. By making $g_i$ highly unprofitable and $h_i$ highly profitable, it follows that it will only be profitable for player~0 to enter a layer, if player~1 moves from $F_i$ to $h_i$.
This corresponds to a set bit. 

Player~0 can force player~1 to move from $F_i$ to $h_i$ by moving from both $d_i$ and $e_i$ to $F_i$.
This is due to the fact the game is a sink game and so the optimal counter-response to a strategy 
cannot be moving into any other cycle than $t$.

The other two remaining nodes of player~0, $b_i$ and $c_i$, build up ladder-like structures as well that connect layers 
with each other, but this time starting from the most-significant to the least-significant bit. The nodes $d_i$ and $e_i$ have direct access to these additional two ladders which will allow them to get reset, corresponding to unsetting
a set bit in the binary counter.

Next, we explain, from a high-level point of view, how the intermediate phases contribute to incrementing the counter. 
At beginning of phase 1 the switches have the following
settings depending
on the value in the binary counter.
All $a_i$ and $e_i$ point upwards if and only if bit $i$ is zero, all $d_i$ point 
to the left, 
and both ladder $b$ and $c$ move down to the 
least significant set bit in the counter. 
We initiate the binary counter at $00 \cdots 001$ which corresponds to the 
to the {\em initial strategy} $\{ a_*^0, b_*^0, c_*^0, d_*^1, e_1^1, e_{*>1}^0 \}$.

\begin{enumerate}
\item 
We apply improving switches in this phase 
s.t.\ all $e_i$ point left, exactly all those $d_i$ point 
to the left that correspond to a set bit in the current or in the next counter state, and update the ladder $b$ s.t.\ it moves down to least significant set bit in the next counter state.
\item In phase~2, we update the ladder $c$ s.t.\ it moves down to least significant set bit in the next counter state.
\item In phase~3, we apply improving switches s.t.\ exactly all those $e_i$ point left that correspond to a set bit in the next counter state.
\item In phase~4, we apply improving switches s.t.\ all $d_i$ point left.
\item In phase~5, we update the $a_i$ ladder to point to the right at exactly those 
bits which are set in the next counter state.
\end{enumerate}
After completing phase~5, the counter has been incremented by one and we can start with phase~1 again.

Phases 1-5 finish when the counter reaches all ones. 
The $a_i$ nodes point to the right, the $d_i$ and $e_i$ nodes point
to the left, the $c$ chain leads down to $g_2$, the $b$ chain leads
down to $g_1$ and all $F_i$ nodes point vertically up. 
This corresponds to the strategy $\{ a_*^1, b_*^0, c_*^0, d_*^1, e_*^1 \}$ with the exception that
$c_2^1$ is chosen instead of $c_2^0$. 
In fact replacing $c_2^1$ by $c_2^0$ is an improving switch and the only such switch.
Therefore it will be chosen by Algorithm~\ref{algorithm: cunningham}.
At this point
Player
0 has no improving switches and loses the game.
We call $\{ a_*^1, b_*^0, c_*^0, d_*^1, e_*^1 \}$
the {\em terminal strategy} of the game. 

Next we specify a total ordering of player 0 edges to be used in
Algorithm~\ref{algorithm: cunningham}.
\begin{displaymath}
\underbrace{\{b_*^*, d_*^0, e_*^1\}}_{\text{Phase 1}} \prec
\underbrace{\{c_*^*\}}_{\text{Phase 2}} \prec
\underbrace{\{e_*^0\}}_{\text{Phase 3}} \prec
\underbrace{\{d_*^1\}}_{\text{Phase 4}} \prec
\underbrace{\{a_*^*\}}_{\text{Phase 5}}
\end{displaymath}
The detailed ordering for every phase is as follows:
\begin{flalign}
&\text{Phase 1}: e_1^1 \prec d_2^0 \prec e_2^1 \prec b_2^1 \prec b_2^0 \prec \ldots \prec d_{n-1}^0 \prec e_{n-1}^1 \prec b_{n-1}^1 \prec b_{n-1}^0 \prec d_n^0 \prec e_n^1 \nonumber \\ 
&\text{Phases 2,3}:c_2^0 \prec c_2^1 \prec \ldots \prec c_n^0 \prec c_n^1,
    ~~~e_1^0 \prec e_2^0 \prec \ldots \prec e_n^0 
\label{ordering} \\
&\text{Phases 4,5}:d_2^1 \prec d_3^1 \prec \ldots \prec d_n^1,~~~~
a_n^1 \prec a_n^0 \prec \ldots \prec a_2^1 \prec a_2^0    \nonumber
\end{flalign}
For each $n \ge 3$ we define $P_n$ to be the sequence of improving switches generated by
Algorithm~\ref{algorithm: cunningham} following this edge ordering, starting at the initial strategy 
and ending at the terminal strategy.
Our goal will be to show that the sequence $P_n$ has exponential length in $n$. 

Before proceeding with the proof let us apply Algorithm~\ref{algorithm: cunningham}
to the example in Figure~\ref{figure: example}. The edge ordering is
\[
e_1^1, d_2^0, e_2^1, b_2^1, b_2^0, d_3^0, e_3^1, c_2^0, c_2^1, c_3^0, c_3^1, e_1^0,
e_2^0, e_3^0, d_2^1, d_3^1, a_3^1, a_3^0, a_2^1, a_2^0
\]

Figure \ref{figure: example}
corresponds to the initial state $001$ of the binary counter,
The current strategy for player 0 is shown by the blues edges, and that
for player 1 by the red edges. 
Improving edges for player 0 are shown
in dotted green 
and the other non-strategy edges are shown in dotted black.
(Coloured edges show in bold on monochromatic printing.)

Note that since bits 2 and 3 of the counter (reading from right to left)
are set to zero the strategy edges for $a_2, a_3, e_2$ and $e_3$ all point upwards.
Each $d_i$ edge points to $F_i$, and the $b$ and $c$ ladders point
to the first bit, which is the least bit set. The improving edges
for player zero are $e_2^1, e_3^1$ and $e_2^1$ is chosen as it comes 
first in order. Since player 0 stands to win on the infinite loop
$(e_2 F_2)$ player 1 counters by changing the strategy on $F_2$ to
point to $h_2$. The first five phases involve nine player 0 moves
and are given in Figure \ref{table: simulation}.
Note in some cases player 1 does not make a response as the position
is still winning for him.
At the end of the nine moves the counter has moved to $010$
and we are back to the settings required to initiate Phase 1.

\begin{figure}[h]
\label{simulation}
\centering
\begin{tabular}[ht]{|c|c|c|c|}
  \hline
  Phase& Improving & Selected & Player 1 \\
       & Edges & Edge & Response \\
  \hline
  \hline
  $1$ & $e_2^1, e_3^1$ & $e_2^1$ & $F_2 h_2$  \\
  $1$ & $a_2^1, b_2^1, c_2^1, e_3^1$ & $b_2^1$ &   \\
  $1$ & $a_2^1, c_2^1, d_3^0, e_3^1$ & $d_3^0$ &   \\
  $1$ & $a_2^1, c_2^1, e_3^1$ & $e_3^1$ & $F_3 d_3$  \\
  $2$ & $a_2^1, c_2^1, d_3^1$ & $c_2^1$ &   \\
  $3$ & $a_2^1, d_3^1, e_1^0, e_3^0$ & $e_3^0$ &   \\
  $3$ & $a_2^1, d_3^1, e_1^0$ & $e_1^0$ &   \\
  $4$ & $a_2^1, d_3^1$ & $d_3^1$ & $F_3 e_3$  \\
  $5$ & $a_2^1, e_3^1$ & $a_2^1$ & $F_1 e_1$  \\
  \hline
\end{tabular}
\caption{{Binary counter moving from 001 to 010}}
\label{table: simulation}
\end{figure}

Continuing in this fashion for a total of 36 improving 
switches by player 0
we arrive at the terminating position where she loses from each node
and has no further improving edges. A complete simulation is given on
the website \cite{online}.

\subsection{Lower Bound Proof}

In this section we prove the fundamental result of the paper.

\begin{theorem} 
\label{lbthm} The sequence $P_n$ of improving switches
followed by Algorithm \ref{algorithm: cunningham} 
in $G_n$ from the initial strategy to the terminal strategy using the ordering 
(\ref{ordering})
has length at least $2^n$ where $G_n$ has size $O(n)$.
\end{theorem}

First, we introduce notation to succinctly describe binary counters. It will be
convenient for us to consider counter configurations with an \emph{infinite}
tape, where unused bits are zero. The set of $n$-bit configurations is
formally defined as $\bitset_n = \{\bit \in \{0,1\}^\infty \mid \forall i>n: \bit_i = 0\}$.

We start with index one, i.e.\ $\bit \in \bitset_n$ is essentially a tuple
$(\bit_{n},\ldots,\bit_1)$, with $\bit_1$ being the least and $\bit_n$ being
the most significant bit. By $\boldzero$, we denote the configuration in which
all bits are zero, and by $\boldone_n$, we denote the configuration in which
the first $n$ bits are one.
We write $\bitset = \bigcup_{n>0} \bitset_n$ to denote the set of all counter
configurations. 

The \emph{integer value} of a $\bit \in \bitset$ is defined as
usual, i.e.\ $|\bit| := \sum_{i>0} \bit_i \cdot 2^{i-1} < \infty$.
For two $\bit, \bit' \in \bitset$, we induce the lexicographic linear ordering
$\bit < \bit'$ by $|\bit| < |\bit'|$.
It is well-known that $\bit \in \bitset \mapsto |\bit| \in \NN$ is a bijection.
For $\bit \in \bitset$ let $\bit^\oplus$ denote the unique $\bit'$
s.t.\ $|\bit'| = |\bit| + 1$.

Given a configuration $\bit$, we access the 
\emph{least unset bit} by $\nub(\bit) = \min \{j \mid \bit_j = 0\}$
and the \emph{least set bit} by $\nsb(\bit) = \min \{j \mid \bit_j = 1\}$.
Let $\bit^\nub$ denote $\bit[\nub(\bit) \mapsto 1]$.

We are now ready to formulate the conditions for strategies that fulfill one of
the five phases along with the improving edges. See Figure~\ref{table: phases}
for a complete description (with respect to a bit configuration $\bit$). We say
that a strategy $\sigma$ is a \emph{phase $p$ strategy with configuration $\bit$} iff
every node is mapped by $\sigma$ to a choice included
in the respective cell of the table.

\begin{figure}[h]
\center
\begin{tabular}[ht]{|l||c|c|c|c|c|c|}
\hline
Phase & 1: $\exists j \geq \nub(\bit), k \in \{1,2,3\}$ & 2 & 3 & 4 & 5: $\exists j \leq \nub(\bit)$ \\
\hline
\hline
$\sigma(a_i)$ & $\bit_i$ & $\bit_i$ & $\bit_i$ & $\bit_i$ & $\twocasefunction{\bit_i}{\bit^\oplus_i}{$i \leq j$}{$i>j$}$ \\
$\sigma(b_i)$ & $\twocaseotherwisefunction{\id_{i = \nub(\bit)}}{\id_{i = \nsb(\bit)}}{$j > i$}$ & $\id_{i = \nub(\bit)}$ & $\id_{i = \nub(\bit)}$ & $\id_{i = \nub(\bit)}$ & $\id_{i = \nub(\bit)}$ \\
$\sigma(c_i)$ & $\id_{i = \nsb(\bit)}$ & $\id_{i = \nsb(\bit)}$ & $\id_{i = \nub(\bit)}$ & $\id_{i = \nub(\bit)}$ & $\id_{i = \nub(\bit)}$ \\
$\sigma(d_i)$ & $\twocaseotherwisefunction{\bit^\nub_i}{1}{$j > i \vee (j=i \wedge k > 1)$}$ & $\bit^\nub_i$ & $\bit^\nub_i$ & 1, $\bit^\nub_i$ & 1 \\
$\sigma(e_i)$ & $\twocaseotherwisefunction{1}{\bit_i}{$j > i \vee (j=i \wedge k > 2)$}$ & 1 & 1, $\bit^\oplus_i$ & $\bit^\oplus_i$ & $\bit^\oplus_i$ \\
\hline	
\end{tabular}
\caption{Policy Phases}
\label{table: phases}
\end{figure}

The following lemma computes an optimal counter-strategy along with the associated
valuations of the nodes $F_i$ given a strategy $\sigma$ belonging to one of the
phases.

\begin{lemma}\label{lemma: valuation of player one nodes}
Let $n \ge 3$ and $\sigma$ be a strategy belonging to one of the phases w.r.t.\ $\bit$.
Let $k = \max (\{i \mid \sigma(a_i) \not= \bit_i^\oplus\} \cup \{1\})$.
Then the following holds: $(\tau_\sigma(F_i), \Xi^{>6}_\sigma(F_i)) =$
\begin{displaymath}
\begin{cases}
	(h_i, \{g_j, h_j \mid \bit_j = 1, j > i\} \cup \{h_i\}) & \text{ if } \sigma(d_i) = \sigma(e_i) = 1 \\
	(e_i, \{g_j, h_j \mid \bit_j = 1\} \cup \{s\}) & \text{ if } \sigma(e_i) = 0 \wedge P=1 \\
	(d_i, \{g_j, h_j \mid \bit^\oplus_j = 1\}) & \text{ if } \sigma(d_i) = 0 \wedge (\sigma(e_i) = 1 \vee P\not=1) \\
	(e_i, \{g_j, h_j \mid \bit^\oplus_j = 1\} \cup \{s\}) & \text{ if } \sigma(e_i) = 0 \wedge \sigma(d_i) = 1 \wedge P\not=1 \wedge i \geq k \\
	(h_i, \{g_j, h_j \mid \bit_j = 1, j > i\} \cup \{h_i\}) & \text{ if } \sigma(e_i) = 0 \wedge \sigma(d_i) = 1 \wedge P\not=1 \wedge i < k = \nub(\bit) \\
	(h_i, \{g_j, h_j \mid \bit^\oplus_j = 1 \vee i < j < k\} \cup \{s,h_i,g_k\}) & \text{ if } \sigma(e_i) = 0 \wedge \sigma(d_i) = 1 \wedge P\not=1 \wedge i < k < \nub(\bit)
\end{cases}
\end{displaymath}
\end{lemma}

\begin{proof}
Let $n \ge 3$ and $\sigma$ be a strategy belonging to one of the phases w.r.t.\ $\bit$.
Let $k = \max (\{i \mid \sigma(a_i) \not= \bit_i^\oplus\} \cup \{1\})$.
\begin{itemize}
\item \emph{Phase 1}:
	First, if $\sigma(d_i) = \sigma(e_i) = 1$, it follows that $\tau_\sigma(F_i) = h_i$
	since we have a sink game.
	Furthermore it follows by definition of phase~1 that
	\begin{displaymath}
	\Xi^{>6}_\sigma(F_i) = \{g_j, h_j \mid \bit_j = 1, j > i\} \cup \{h_i\}
	\end{displaymath}
	
	Otherwise, $i \geq \nub(\bit)$. It follows by definition of phase~1 that
	\begin{flalign*}
	\Xi^{>6}_\sigma(s) &= \{g_j, h_j \mid \bit_j = 1\} \cup \{s\} \\
	\Xi^{>6}_\sigma(h_i) &= \{g_j, h_j \mid \bit_j = 1, j > i\} \cup \{h_i\}
	\end{flalign*}
	and if $\sigma(d_i) = 0$ that
	\begin{displaymath}
	\Xi^{>6}_\sigma(s) = \{g_j, h_j \mid \bit^\oplus_j = 1\}
	\end{displaymath}
	
	Second, if $\sigma(d_i) = 1$ and $\sigma(e_i) = 0$, it follows that
	$e_i \prec_\sigma h_i$ and therefore $\tau_\sigma(F_i) = e_i$ and
	\begin{displaymath}
	\Xi^{>6}_\sigma(F_i) = \{g_j, h_j \mid \bit_j = 1\} \cup \{s\}
	\end{displaymath}
	
	Third, if $\sigma(d_i) = 0$ and $\sigma(e_i) = 1$, it follows that
	$d_i \prec_\sigma h_i$ and therefore $\tau_\sigma(F_i) = d_i$ and
	\begin{displaymath}
	\Xi^{>6}_\sigma(F_i) = \{g_j, h_j \mid \bit^\oplus_j = 1\}
	\end{displaymath}
	
	Fourth, if $\sigma(d_i) = \sigma(e_i) = 0$, it follows that
	$e_i \prec_\sigma d_i \prec_\sigma h_i$ and therefore $\tau_\sigma(F_i) = e_i$ and
	\begin{displaymath}
	\Xi^{>6}_\sigma(F_i) = \{g_j, h_j \mid \bit_j = 1\} \cup \{s\}
	\end{displaymath}

\item \emph{Phase 2}:
	First, if $\bit_i^\nub = 1$, it follows that $\sigma(d_i) = \sigma(e_i) = 1$ by
	definition of phase~2. Since we have a sink game, it follows that $\tau_\sigma(F_i) = h_i$.
	Furthermore it follows by definition of phase~2 that
	\begin{displaymath}
	\Xi^{>6}_\sigma(F_i) = \{g_j, h_j \mid \bit_j = 1, j > i\} \cup \{h_i\}
	\end{displaymath}
	
	Second, if $\bit_i^\nub = 0$, if follows that $i > \nub(\bit)$ and from
	definition of phase~2 that $\sigma(d_i) = 0$, $\sigma(e_i) = 1$, and
	\begin{flalign*}
	\Xi^{>6}_\sigma(h_i) &= \{g_j, h_j \mid \bit_j = 1, j > i\} \cup \{h_i\} \\
	\Xi^{>6}_\sigma(b_{i-1}) &= \{g_j, h_j \mid \bit^\oplus_j = 1\}
	\end{flalign*}
	Hence, $b_{i-1} \prec_\sigma h_i$ and therefore $\tau_\sigma(F_i) = d_i$
	and
	\begin{displaymath}
	\Xi^{>6}_\sigma(F_i) = \{g_j, h_j \mid \bit^\oplus_j = 1\}
	\end{displaymath}

\item \emph{Phase 3 \& 4}:
	First, if $\bit_i^\oplus = 1$, it follows that $\sigma(d_i) = \sigma(e_i) = 1$ by
	definition of phase~3 and phase~4. Since we have a sink game, it follows that $\tau_\sigma(F_i) = h_i$.
	Furthermore it follows by definition of phase~3 and phase~4 that
	\begin{displaymath}
	\Xi^{>6}_\sigma(F_i) = \{g_j, h_j \mid \bit_j = 1, j > i\} \cup \{h_i\}
	\end{displaymath}
	
	Second, if $\bit_i^\nub = 0$, it follows that $i > \nub(\bit)$. By definition
	of phase~3 and phase~4, we have
	\begin{flalign*}
	\Xi^{>6}_\sigma(h_i) &= \{g_j, h_j \mid \bit_j = 1, j > i\} \cup \{h_i\} \\
	\Xi^{>6}_\sigma(b_{i-1}) &= \{g_j, h_j \mid \bit^\oplus_j = 1\} \\
	\Xi^{>6}_\sigma(s) &=  \{g_j, h_j \mid \bit^\oplus_j = 1\} \cup \{s\}
	\end{flalign*}
	It follows that $b_{i-1} \prec_\sigma s \prec_\sigma h_i$ and hence
	$\tau_\sigma(F_i) = d_i$ and
	\begin{displaymath}
	\Xi^{>6}_\sigma(F_i) = \{g_j, h_j \mid \bit^\oplus_j = 1\}
	\end{displaymath}
	if $\sigma(d_i) = 0$, and
	$\tau_\sigma(F_i) = e_i$ and
	\begin{displaymath}
	\Xi^{>6}_\sigma(F_i) = \{g_j, h_j \mid \bit^\oplus_j = 1\} \cup \{s\}
	\end{displaymath}
	if $\sigma(d_i) = 1$.
	
	Third, if $i < \nub(\bit) = k$, we show by backward induction on $i$ that
	$\tau_\sigma(F_i) = h_i$, which implies that
	\begin{displaymath}
	\Xi^{>6}_\sigma(F_i) = \{g_j, h_j \mid \bit_j = 1, j > i\} \cup \{h_i\}
	\end{displaymath}
	By induction hypothesis or by considering the base case $i = k - 1$ directly,
	we have
	\begin{displaymath}
	\Xi^{>6}_\sigma(h_{i}) = \{g_j, h_j \mid \bit_j = 1, j > i\} \cup \{h_{i}\}
	\end{displaymath}
	and therefore $h_{i} \prec_\sigma s$, implying $\tau_\sigma(F_{i}) = h_{i}$.
		
\item \emph{Phase 5}:
	First, if $\bit_i^\oplus = 1$, it follows that $\sigma(d_i) = \sigma(e_i) = 1$ by
	definition of phase~5. Since we have a sink game, it follows that $\tau_\sigma(F_i) = h_i$.
	Furthermore it follows by definition of phase~5 that
	\begin{displaymath}
	\Xi^{>6}_\sigma(F_i) = \{g_j, h_j \mid \bit_j = 1, j > i\} \cup \{h_i\}
	\end{displaymath}
	
	Second, if $\bit_i^\oplus = 0$ and $i \geq k$, it follows by definition of phase~5
	that
	\begin{flalign*}
	\Xi^{>6}_\sigma(h_i) &= \{g_j, h_j \mid \bit_j^\oplus = 1, j > i\} \cup \{h_i\} \\
	\Xi^{>6}_\sigma(s) &= \{g_j, h_j \mid \bit^\oplus_j = 1\} \cup \{s\}
	\end{flalign*}
	It follows that $s \prec_\sigma h_i$ and hence $\tau_\sigma(F_i) = e_i$ and
	\begin{displaymath}
	\Xi^{>6}_\sigma(F_i) = \{g_j, h_j \mid \bit^\oplus_j = 1\} \cup \{s\}
	\end{displaymath}
	
	Third, if $i < \nub(\bit) = k$, we show by backward induction on $i$ that
	$\tau_\sigma(F_i) = h_i$, which implies that
	\begin{displaymath}
	\Xi^{>6}_\sigma(F_i) = \{g_j, h_j \mid \bit_j = 1, j > i\} \cup \{h_i\}
	\end{displaymath}
	By induction hypothesis or by considering the base case $i = k - 1$ directly,
	we have
	\begin{displaymath}
	\Xi^{>6}_\sigma(h_{i}) = \{g_j, h_j \mid \bit_j = 1, j > i\} \cup \{h_{i}\}
	\end{displaymath}
	and therefore $h_{i} \prec_\sigma s$, implying $\tau_\sigma(F_{i}) = h_{i}$.

	Fourth, if $i < k < \nub(\bit)$, we show by backward induction on $i$ that
	$\tau_\sigma(F_i) = h_i$, which implies that
	\begin{displaymath}
	\Xi^{>6}_\sigma(F_i) = \Xi^{>6}_\sigma(s) \cup \{g_{j+1}, h_j \mid i \leq j < k\}
	\end{displaymath}
	By induction hypothesis or by considering the base case $i = k - 1$ directly,
	we have
	\begin{displaymath}
	\Xi^{>6}_\sigma(h_{i}) = \Xi^{>6}_\sigma(s) \cup \{g_{j+1}, h_j \mid i < j < k\} \cup \{h_i, g_{i+1}\}
	\end{displaymath}
	and therefore $h_{i} \prec_\sigma s$, implying $\tau_\sigma(F_{i}) = h_{i}$.
	
\end{itemize}
\qed
\end{proof}

Figure~\ref{table: improving switches} specifies the sets of improving switches 
for each phase $p$. It should be read as follows: an edge $e$ is included in the
set of improving switches $I(\sigma)$ iff $e \not\in \sigma$ and the 
condition holds that is specified in the respective cell. If a cell contains a
question mark, we do not specify whether the edge is included in the set.

\begin{figure}[h]
\center
\begin{tabular}[ht]{|l||c|c|c|c|c|c|}
\hline
Phase & 1 & 2 & 3 & 4 & 5 \\
\hline
\hline
$a_i^1 \in I(\sigma)$ & ? & ? & ? & ? & $i = \nub(\bit)$ \\
$a_i^0 \in I(\sigma)$ & ? & ? & ? & ? & $i<\nub(\bit) \wedge \sigma(a_{i+1}) = \bit_i^\oplus$ \\
$b_i^1 \in I(\sigma)$ & $i=\nub(\bit) \wedge \sigma(e_i)=1$ & ? & ? & ? & ? \\
$b_i^0 \in I(\sigma)$ & $i>\nub(\bit) \wedge \sigma(e_{\nub(\bit)})=1$ & ? & ? & ? & ? \\
$c_i^1 \in I(\sigma)$ & ? & $i=\nub(\bit)$ & ? & ? & ? \\
$c_i^0 \in I(\sigma)$ & ? & $i\not=\nub(\bit)$ & ? & ? & ? \\
$d_i^1 \in I(\sigma)$ & ? & ? & ? & Yes & ? \\
$d_i^0 \in I(\sigma)$ & $\sigma(b_{\nub(\bit)})=1 \wedge \sigma(b_{\nsb(\bit)}) = 0$ & ? & ? & ? & ? \\
$e_i^1 \in I(\sigma)$ & Yes & ? & ? & ? & ? \\
$e_i^0 \in I(\sigma)$ & ? & ? & $\bit^\oplus_i = 0$ & ? & ? \\
\hline	
\end{tabular}
\caption{Improving Switches}
\label{table: improving switches}
\end{figure}

\noindent We finally arrive at the following main lemma describing the improving switches.

\begin{lemma}\label{lemma: improving switches}
Let $n \ge 3$.
The improving switches from policies that belong to the phases in Figure~\ref{table: phases}
are as specified in Figure~\ref{table: improving switches}.
\end{lemma}

\begin{proof}
Let $n \ge 3$ and $\sigma$ be a strategy belonging to one of the phases w.r.t.\ $\bit$.
\begin{itemize}
\item \emph{Phase 1}:
	It follows from Lemma~\ref{lemma: valuation of player one nodes} that
	\begin{flalign*}
	\Xi^{>6}_\sigma(g_i) &= \begin{cases}
		\{g_j, h_j \mid \bit_j = 1, j > i\} \cup \{g_i, h_i\} & \text{ if } \sigma(d_i) = \sigma(e_i) = 1 \\
		\{g_j, h_j \mid \bit_j = 1\} \cup \{g_i, s\} & \text{ if } \sigma(e_i) = 0\\
		\{g_j, h_j \mid \bit^\oplus_j = 1\} \cup \{g_i\} & \text{ if } \sigma(d_i) = 0 \wedge \sigma(e_i) = 1\\
	\end{cases}
	\end{flalign*}
	Since $\tau_\sigma(F_i) = e_i$ if $\sigma(e_i) = 0$, it immediately follows that $e_i^1$ is an improving switch.
	It can furthermore be easily observed that $d_i^0$ is an improving switch for $\bit_i = 0$, $i > \nsb(\bit)$ iff
	$\sigma(e_{\nsb(\bit)}) = 1$ and $\sigma(b_{\nsb(\bit)}) = 1$.

\item \emph{Phase 2}:
	It follows from Lemma~\ref{lemma: valuation of player one nodes} that
	\begin{displaymath}
	\Xi^{>6}_\sigma(g_i) = \begin{cases}
		\{g_j, h_j \mid \bit_j = 1, j > i\} \cup \{g_i, h_i\} & \text{ if } \bit_i^\nub = 1\\
		\{g_j, h_j \mid \bit^\oplus_j = 1\} \cup \{g_i\} & \text{ if } \bit_i^\nub = 0
	\end{cases}
	\end{displaymath}
	
	If $\bit_1 = 0$ and $\bit_2 = 1$, it follows that
	\begin{displaymath}
	\Xi^{>6}_\sigma(c_i) = \{g_j, h_j \mid \bit_j = 1\}
	\end{displaymath}
	It immediately follows that $c_2^0$ is the only improving switch w.r.t.\ $c$.
	
	Otherwise, if $\bit_1 = 0$ and $\bit_2 = 0$, it follows that
	\begin{displaymath}
	\Xi^{>6}_\sigma(c_i) = \begin{cases}
		\{g_j, h_j \mid \bit_j = 1\} & \text{ if } i \geq \nsb(\bit) \\
		\{g_j, h_j \mid \bit^\oplus_j\} & \text{ if } i < \nsb(\bit)
	\end{cases}
	\end{displaymath}
	It immediately follows that $c_{\nsb(\bit)}^0$ is the only improving switch w.r.t.\ $c$.

	Otherwise, if $\bit_1 = 1$, it follows that
	\begin{displaymath}
	\Xi^{>6}_\sigma(c_i) =
		\{g_j, h_j \mid \bit_j = 1\}
	\end{displaymath}
	Hence, it follows that $c_{\nub(bit)}^1$ is the only improving switch w.r.t.\ $c$.

\item \emph{Phase 3}:
	Let $i$ s.t.\ $\sigma(e_i) = 1$. It follows from Lemma~\ref{lemma: valuation of player one nodes}
	that
	\begin{flalign*}
	\Xi^{>6}_\sigma(F_i) &= \begin{cases}
		\{g_j, h_j \mid \bit_j = 1, j > i\} \cup \{h_i\} & \text{ if } \sigma(d_i) = 1\\
		\{g_j, h_j \mid \bit^\oplus_j = 1\} & \text{ if } \sigma(d_i) = 0
	\end{cases} \\
	\Xi^{>6}_\sigma(s) &= \{g_j, h_j \mid \bit^\oplus_j = 1\} \cup \{s\}
	\end{flalign*}
	Hence, we have $F_i \prec_\sigma s$.

\item \emph{Phase 4}:
	Let $i$ s.t.\ $\sigma(d_i) = 0$. It follows from Lemma~\ref{lemma: valuation of player one nodes}
	that $\tau_\sigma(F_i) = d_i$ and hence $\Xi_\sigma(F_i) = \{F_i\} \cup \Xi_\sigma(d_i)$,
	i.e.\ $\sigma(d_i) \prec_\sigma F_i$.

\item \emph{Phase 5}:
	Let $k = \max (\{i \mid \sigma(a_i) \not= \bit_i^\oplus\} \cup \{1\})$.
	It follows from Lemma~\ref{lemma: valuation of player one nodes}
	that
	\begin{displaymath}
	\Xi^{>6}_\sigma(g_i) = \begin{cases}
		\{g_j, h_j \mid \bit_j = 1, j > i\} \cup \{g_i, h_i\} & \text{ if } \bit^\oplus_i = 1 \vee (i < k=\nub(\bit)) \\
		\{g_j, h_j \mid \bit^\oplus_j = 1\} \cup \{g_i, s\} & \text{ if } \bit^\oplus_i = 0 \wedge i\geq k\\
		\{g_j, h_j \mid \bit^\oplus_j = 1 \vee i<j<k\} \cup \{s, g_k, g_i, h_i\} & \text{ if } \bit^\oplus_i = 0 \wedge i< k<\nub(\bit)
	\end{cases}
	\end{displaymath}
	Hence, it follows that if $k > 1$, the only improving switch is either $a_k^1$ or $a_k^0$ w.r.t.\ $a$.
	
\end{itemize}
\qed
\end{proof}

We are now ready to formulate our main lemma describing the transitioning from
an initial phase 1 strategy corresponding to $\bit$ to a successor initial
phase 1 strategy corresponding to $\bit^\oplus$, complying with the given
ordering selection.

\begin{lemma}\label{lemma: policy proceeding}
Let $\sigma$ be a phase~1 strategy with configuration $\boldzero_n < \bit < \boldone_n$. Let
$z$ be an edge with $z \preceq e_1^1$ or
$a_n^1 \preceq e$.
Then, there is a phase~1 strategy $\sigma'$ with configuration $\bit^\oplus$ and
an edge $z'$ with $z' \preceq e_1^1$ or
$a_n^1 \preceq z'$ s.t.\
$(\sigma, z) \leadsto_\prec (\sigma', z')$.
\end{lemma}

\begin{proof}
The proof of the lemma is ultimately based on the five phases described in
Figure~\ref{table: phases}, the corresponding improving switches given in Figure~\ref{table: improving switches}
(proven correct in Lemma~\ref{lemma: improving switches}) and the introduced selection ordering.

We prove the lemma by outlining the complete sequence of switches that are applied
to $\sigma$ in order to obtain $\sigma'$ (we do not explicitly describe the intermediate
strategies which can be derived by applying all mentioned switches up to that point).

Let $i_1, \ldots, i_k$ be the complete sequence of ascending indices s.t.\ $\bit_{i_j} = 0$
for $1\leq j \leq k$. The following holds:
\begin{align*}
  &\stackrel{\text{P1}}{\leadsto} e_{i_1}^1 \leadsto b_{i_1}^1 \leadsto d_{i_2}^0 \leadsto e_{i_2}^1 \leadsto d_{i_3}^0 \leadsto e_{i_3}^1 \leadsto \ldots \leadsto d_{i_k}^0 \leadsto e_{i_k}^m1 \\
  &\stackrel{\text{P2}}{\leadsto} c_{i_1}^1 \\
  &\stackrel{\text{P3}}{\leadsto} \{e_i^0 \mid \bit^\oplus = 0\} \\
  &\stackrel{\text{P4}}{\leadsto} \{d_i^1 \mid \bit^\nsb = 0\} \\
  &\stackrel{\text{P5}}{\leadsto} a_{i_1}^* \leadsto a_{i_1 - 1}^* \leadsto \ldots \leadsto a_{2}^*  
\end{align*}
\qed
\end{proof}

It follows immediately that the parity game provided here indeed simulate a binary
counter by starting with the designated initial strategy and the $\prec$-minimal edge.
This completes the proof of Theorem \ref{lbthm}.

\subsection{Application to Other Games}

Theorem \ref{lbthm} can be applied to a variety of other games using known connections
between these games. For completeness we give these results here. However we will
not give definitions of the games concerned, referring instead to the relevant literature,
as they will not be needed in the rest of the paper.

Parity games can be reduced to mean payoff games \cite{puri/phd}, mean payoff games to
discounted payoff games, and the latter ones to turn-based stochastic games
\cite{zwickpaterson/1996}. Friedmann showed \cite{FriedmannExpBound09}
that the strategy improvement algorithm for payoff and stochastic games, when
applied to the reduced game graphs of sink parity games, behaves exactly the
same on the reduced games. In other words, the strategies in both the original
game and the reduced game graphs coincide as well as the associated sets of
improving switches.

\begin{theorem}[\cite{FriedmannExpBound09}]
Let~$G$ be a sink parity game. Discrete strategy improvement requires the same number of iterations to solve~$G$
as strategy improvement for the induced payoff games as well as turn-based stochastic games to solve the respective
game $G'$. The game $G'$ is induced by applying the standard reduction from~$G$ to the 
respective game class,
assuming that the improvement rule solely depends on the combinatorial
valuation-ordering of the improving edges.
\end{theorem}

By this, we can conclude that the exponential lower bound for Cunningham's rule on
parity games presented here also applies to payoff and turn-based stochastic
games. We will see in the next section that it also applies to Markov decision
processes, which can be seen as a one-player version of turn-based stochastic
games.

\begin{corollary}
There is a family of mean payoff games, discounted payoff
games resp.\ turn-based stochastic games on which number of improving steps
performed by Algorithm \ref{algorithm: cunningham} is at least $2^n$, where 
the size of the games are $O(n)$.
\end{corollary}

\section{Markov Decision Process Policy Iteration Lower Bound}
\label{mdp}
Markov decision processes (MDPs) provide a mathematical model for
sequential decision making under uncertainty. They are employed
to model stochastic optimization problems in various areas ranging
from operations research, machine learning, artificial intelligence,
economics and game theory. For an in-depth coverage of MDPs, see the books of
Howard \cite{Howard60}, Derman \cite{Derman72}, Puterman
\cite{Puterman94} and Bertsekas \cite{Bertsekas01}.

\subsection{Markov Decision Processes and Policy Iteration}
Formally, an MDP is defined by its \emph{underlying graph} $G{=}(V_0,V_R,E_0,E_R,r,p)$. Here, $V_0$ is the set of vertices (states) operated by the controller, also known as \emph{player~$0$}, and $V_R$ is a set of \emph{randomization} vertices corresponding to the probabilistic actions of the MDP. We let $V=V_0\cup V_R$. The edge set $E_0\subseteq V_0\times V_R$ corresponds to the actions available to the controller.
The edge set $E_R \subseteq V_R\times V_0$ corresponds to the probabilistic transitions associated with each action.
The function $r:E_0\to\mathbb{R}$ is the immediate reward function. The function $p:E_R\to [0,1]$ specifies the transition probabilities. For every $u\in V_R$, we have $\sum_{v:(u,v)\in E_R} p(u,v) = 1$, i.e., the probabilities of all edges emanating from each vertex of $V_R$ sum up to~$1$.
%

A policy $\sigma$ is a function $\sigma:V_0\to V$ that selects for each vertex $u\in V_0$ a target node $v$ corresponding to an edge $(u,v) \in E_0$, i.e.\ $(u,\sigma(u)) \in E_0$. We assume that each vertex $u\in V_0$ has at least one outgoing edge). There are several \emph{objectives} for MDPs and we consider the \emph{expected total reward objective} here. 
The \emph{values} $\mdpval{\sigma}(u)$ of the vertices under~$\sigma$ are defined as the unique solutions of the following set of linear equations:
\begin{displaymath}
\mdpval{\sigma}(u) = \begin{cases}
\mdpval{\sigma}(v) + r(u,v) & \text{if $u\in V_0$ and $\sigma(u)=v$}\\
\sum_{v:(u,v)\in E_R} p(u,v)\, \mdpval{\sigma}(v) & \text{if $u\in V_R$}
\end{cases}
\end{displaymath}
together with the condition that $\mdpval{\sigma}(u)$ sum up to~$0$ on each irreducible recurrent class of the Markov chain defined by~$\sigma$.

All MDPs considered in this paper satisfy the
\emph{unichain} condition (see \cite{Puterman94}). It
states that the Markov chain obtained from each policy $\sigma$ has a single
irreducible recurrent class. 

This condition implies, in particular, that all vertices have the same value. It is not difficult to check that $\mdpval{\sigma}(u)$ is indeed the expected reward per turn, when the process starts at~$u$ and policy~$\sigma$ is used. The potentials $\mdppot{\sigma}(u)$ represent \emph{biases}. Loosely speaking, the expected reward after~$N$ steps, when starting at~$u$ and following~$\sigma$, and when~$N$ is sufficiently large, is about $N\,\mdpval{\sigma}(u)+\mdppot{\sigma}(u)$.

Howard's \cite{Howard60} \emph{policy iteration} algorithm is the most widely
used algorithm for solving MDPs. It is closely related to the simplex method, which provides
a practical way to solve such problems. Nevertheless in the worst case,
Melekopoglou and Condon \cite{MC94} showed that the
simplex method with the smallest index pivot rule
needs an exponential number of iterations to compute
an optimal policy for a specific MDP problem regardless
of discount factors.

Although most MDPs can be solved in polynomial time using the interior point method the search continues
for a strongly polynomial time method.
Post and Ye \cite{PY12} have recently made progress in this direction by showing that the simplex method with the
greedy pivot rule terminates in at most $O(m^3 n^2 log^2 m)$
pivot steps when discount factors are uniform, or in at
most $O(m^5 n^3 log^2 m)$
pivot steps with non-uniform discounts.
No such results have been proved for the
policy iteration method. As history based pivot rules provide good candidates for subexponential
time behaviour it is important to analyze their worst case performance for policy improvement algorithms.

As is the case for parity games, the policy iteration algorithm starts with some initial policy $\sigma_0$
and generates an improving sequence $\sigma_0,\sigma_1,\ldots,\sigma_N$ of
policies, ending with an optimal policy $\sigma_N$. In each iteration the algorithm
first \emph{evaluates} the current policy $\sigma_i$, by computing the values
$\mdpval{\sigma_i}(u)$  of all vertices. An edge $(u,v')\in E_0$, such that
$\sigma_i(u) \not= v'$ is then said to be an \emph{improving switch} if and only
if either $\mdpval{\sigma_i}(v')>\mdpval{\sigma_i}(u)$.
Given a policy $\sigma$, we again denote the \emph{set of improving switches} by
$I_\sigma$.

A crucial property of policy iteration is that $\sigma$ is an optimal policy if
and only if there are no improving switches with respect to it (see, e.g.,
\cite{Howard60}, \cite{Puterman94}).
Furthermore, if $(u,v') \in I_\sigma$ is an improving switch w.r.t.\ $\sigma$,
and $\sigma'$ is defined as $\sigma[(u,v')]$ (i.e., $\sigma'(u) = v'$ and
$\sigma'(w) = \sigma(w)$ for all $w\not= u$), then $\sigma'$ is \emph{strictly}
better than $\sigma$, in the sense that for every $u\in V_0$, we have
$\mdpval{\sigma'}(u)\geq\mdpval{\sigma}(u)$, with a strict inequality for at
least one vertex $u\in V_0$.

\begin{lemma}\label{lemma: increasing valuation mdp}
Let $\sigma$ be a policy and $(v,w)$ be a $\sigma$-improving switch. Let $\sigma' = \sigma[v \mapsto w]$.
Then the following holds:
\begin{enumerate}
  \item $\mdpval{\sigma'}(u) \geq \mdpval{\sigma}(u)$ for all $u \in V$,
  \item $\mdpval{\sigma'}(v) > \mdpval{\sigma}(v)$, and particularly
  \item $\sum_{u \in V_0} \mdpval{\sigma'}(u) > \sum_{u \in V_0} \mdpval{\sigma}(u)$.
\end{enumerate}
\end{lemma}

\subsection{Lower Bound Construction}

We relate the description of the lower bound construction for Markov decision
processes closely to the construction of the parity games. For that reason, we
relax our definition of MDPs such that it corresponds almost directly to parity
games.

As defined in the previous subsection, the underlying graph $G$ is bipartite.
However one can relax this condition and allow edges from $V_0$ to $V_0$ 
and from $V_R$ to $V_R$. It suffices to subdivide these edges by inserting
a node belonging to player 1 or player 0, respectively, with out-degree 1 and
no priority. This leads to the following definition.

A \emph{relaxed MDP} is a tuple $M=(V, V_0, V_R, E, E_0, E_R, r, p)$, where $V =
V_0 \cup V_R$, $E \subseteq V \times V$, $E_0 = E \cap (V_0 \times V)$, $E_R =
E \cap (V_R \times V)$, $r: E \mapsto \RR$ and $p: E_R \rightarrow [0, 1]$ with
$\sum_{w \in vE} p(v,w) = 1$ for all $v \in V_R$.

Relaxed MDPs allow us to show the close relationship between the original parity
games and the corresponding MDPs. 
\begin{enumerate}
  \item Edges $(v,w) \in E_0 \cap V_0 \times V_0$ can be realized by adding a
  randomization node $(v,w)$ and by replacing the edge $(v,w)$ with new edges
  $(v, (v,w))$ and $((v,w), w)$. The outgoing edge from $(v,w)$ obviously has
  probability 1.
  \item Edges $(v,w) \in E_R \cap V_R \times V_R$ can be realized by adding a
  player 0 node $(v,w)$ and by replacing the edge $(v,w)$ with new edges $(v,
  (v,w))$ and $((v,w), w)$. Note that player 0 does not obtain new choices
  by adding this node since the out-degree is one.
  \item Randomization edges are not allowed to have rewards
  in MDPs.
  Hence we insert, for every outgoing edge of a randomization node with reward,
  a new node of player 0 connected to a new node of the randomizer connected to the original target node. We push
  the reward to the new player 0 edge.
  \item Since we consider the expected total reward here, adding new intermediate nodes
  does not change the value of the nodes.
\end{enumerate}

For each $n \ge 3 $ we define the underlying graph $M_n = (V,V_0,V_R,E,E_0,E_R,r,p)$ of a relaxed
MDP as shown schematically in Figure~\ref{figure: mdp}
which the reader is invited to compare with Figure~\ref{figure: example}. 

More formally:
\begin{flalign*}
V_0 \;:=\;& \{a_i, c_i, d_i \mid 1 < i \leq n\} \cup \{b_i \mid 1 < i < n\} \cup
\{e_i \mid 1 \leq i \leq n\}  \\
V_R \;:=\;& \{F_i \mid 1 \leq i \leq n\} \cup \{g_i,h_i \mid 1 \leq i \leq n\}  \cup \{s,t\}
\end{flalign*}

With $M_n$, we associate a large number $N \in \NN$ and a small number
$0 < \varepsilon$. We require $N$ to be at least as large as the
number of nodes with priorities, i.e.\ $N \geq 2 n$ and
$\varepsilon^{-1}$ to be significantly larger than the largest
occurring priority induced reward, i.e.\ $\varepsilon \leq
\frac{1}{2n}$.

Some of the vertices are assigned integer \emph{priorities}. If a
vertex~$v$ has priority~$\Omega(v)$ assigned to it, then a reward of
$\pri{v}=(-N)^{\Omega(v)}$ is \emph{added} to all edges emanating from~$v$.
This idea of using priorities is inspired by the reduction from parity games
to mean payoff games, see \cite{puri/phd}.

Figure~\ref{table: mdp lower bound edges} defines the edge sets, the
probabilities and the priorities of $M_n$. For convenience of notation, we
identify the node names $a_{n+1}$ with $t$,
$b_1$ with $g_1$, and
$c_1$ with $g_1$. Explicit constructions for small $n$ are available online \cite{online}.

\begin{figure}[h]
\centering
\begin{tabular}[ht]{|c||c|c|}
  \hline
  Node & Successors & Probability \\
  \hline
  \hline
  $F_i$ & $h_i$ & $\varepsilon$ \\
  $F_i$ & $d_i$ & $0.5 \cdot (1-\varepsilon)$ \\
  $F_i$ & $e_i$ & $0.5 \cdot (1-\varepsilon)$ \\
  \hline
  \hline
  Node & Successors & Priority \\
  \hline
  \hline
  $g_i$ & $F_i$ & $2 \cdot i - 1$ \\
  $h_i$ & $a_{i+1}$ & $2 \cdot i$ \\
  $s$ & $c_n$ & $0$ \\
  \hline
\end{tabular}
\begin{tabular}[ht]{|c||c|}
  \hline
  Node & Successors \\
  \hline
  \hline
  $a_i$ & $g_i$, $a_{i+1}$ \\		
  $b_i$ & $g_i$, $b_{i-1}$	\\	
  $c_i$ & $g_i$, $c_{i-1}$		\\
  $d_i$ & $F_i$, $b_{i-1}$ \\
  $e_i$ & $F_i$, $s$ \\
  $t$ & $t$ \\
  \hline
\end{tabular}
\caption{{MDP Lower Bound Graph}}
\label{table: mdp lower bound edges}
\end{figure}

\begin{figure}[ht!]
\center
\includegraphics[scale=0.58]{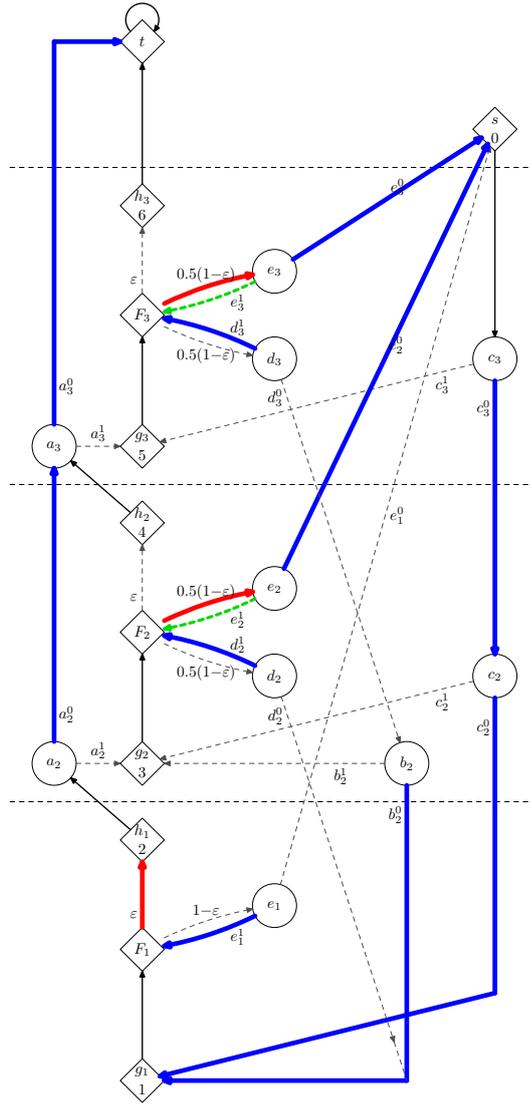}
\caption{(Relaxed) Markov Decision Process Lower Bound Graph}
\label{figure: mdp}
\end{figure}

In comparing Figure~\ref{figure: mdp}
with Figure~\ref{figure: example} the connections between MDPs and parity games becomes clear.
Blue edges show
respectively the current policy and current policy for player 0, and red edges for player 1.
Improving edges for player 0 are shown in dotted green and
the other current non-policy and non-strategy edges are shown in dotted black.
(Coloured edges show in bold on monochromatic printing.)
Whereas Figure~\ref{figure: example} shows the initial strategy for parity game $G_3$,
Figure~\ref{figure: mdp} shows the initial policy for the MDP $M_3$.
The sequence $P_n$ of improving switches
we constructed in $G_n$ starting at the initial strategy and following
Algorithm~\ref{algorithm: cunningham} using the ordering (\ref{ordering})
will be shown to correspond to an identical sequence of improving switches
in $M_n$ using policy iteration and the same ordering.

\begin{lemma}\label{lemma: mdp unichain}
The Markov chains obtained by any policy
reach the sink $t$ almost surely (i.e.\ the sink $t$ is the
single irreducible recurrent class).
\end{lemma}

It is not too hard to see that the absolute value of all nodes corresponding to
policies are bounded by $\varepsilon^{-1}$. More formally we have:

\begin{lemma}\label{lemma: priorities work as in parity games}
Let $P = \{s, g_*, h_*\}$ be the
set of nodes with priorities. For a subset $S \subseteq P$, let
$\sum(S) = \sum_{v \in S} \pri{v}$. For non-empty subsets $S \subseteq P$,
let $v_S \in S$ be the node with the largest priority in $S$.
\begin{enumerate}
\item $|\sum(S)| < N + 1$ and $\varepsilon \cdot |\sum(S)| < 1$ for every subset
$S \subseteq P$, and
\item $|v_S| < |v_{S'}|$ implies $|\sum(S)| < |\sum(S')|$ for non-empty subsets $S,S' \subseteq P$.
\end{enumerate}
\end{lemma}

\subsection{Lower Bound Proof}

In this section we prove the following theorem.
The {\em initial policy}, $\{ a_*^0, b_*^0, c_*^0, d_*^1, e_*^0 \}$, corresponds exactly 
to the initial strategy for the parity game $G_n$.

\begin{theorem}
\label{theorem: mdp} The sequence $P_n$ of improving switches followed by policy iteration 
in $M_n$ from the initial policy to the terminal policy using the ordering
(\ref{ordering})
has length at least $2^n$, where $M_n$ has size $O(n)$.
\end{theorem}

We will show that the sequence $P_n$ in Theorem \ref{theorem: mdp} is in one-to-one correspondence with
the sequence $P_n$ in Theorem \ref{lbthm}.
We first show that Figure~\ref{table: improving switches} indeed specifies the sets of improving switches 
for each phase $p$. 

\begin{lemma}
Let $n > 1$.
The improving switches from policies that belong to the phases in Figure~\ref{table: phases}
are as specified in Figure~\ref{table: improving switches}.
\end{lemma}

\begin{proof}
Let $n > 1$ and $\sigma$ be a policy belonging to one of the phases w.r.t.\
$\bit$. Let $\mu = \mu(\bit)$. Define $T_i = \pri{h_i} + \sum_{j > i,\bit_j = 1} (\pri{g_j} + \pri{h_j})$ and $S_i = \pri{g_i} + T_i$.
\begin{itemize}

\item \emph{Phase 1}:
	The following holds:
	\begin{flalign*}
		\mdpval{\sigma}(g_i) &= \begin{cases}
			S_i & \text{ if } \sigma(d_i) = 1, \sigma(e_i) = 1 \\
			\varepsilon \cdot T_i + \pri{g_i} + (1-\varepsilon) \cdot (\pri{s} + S_\mu) & \text{ if } \sigma(d_i) = 1, \sigma(e_i) = 0 \\
			\varepsilon \cdot T_i + \pri{g_i} + (1-\varepsilon) \cdot S_\mu & \text{ if } \sigma(d_i) = 0, \sigma(e_i) = 1 \\
			\varepsilon \cdot T_i + \pri{g_i} + \frac{1-\varepsilon}2 \cdot (\pri{s} + 2\cdot S_\mu) & \text{ if } \sigma(d_i) = 0, \sigma(e_i) = 0
		\end{cases} \\
		\mdpval{\sigma}(e_i) &= \begin{cases}
			\pri{s} + S_\mu & \text{ if } \sigma(e_i) = 0
		\end{cases}
	\end{flalign*}
	It is easy to see that $e_i^1$ are improving switches as $T_i > S_\mu$.
	It can furthermore be easily observed that $d_i^0$ is an improving switch for $\bit_i = 0$, $i > \nsb(\bit)$ iff
	$\sigma(e_{\nsb(\bit)}) = 1$ and $\sigma(b_{\nsb(\bit)}) = 1$.

\item \emph{Phase 2}:
	Similar to phase 5.

\item \emph{Phase 3}:
	Similar to phase 4.

\item \emph{Phase 4}:
	The following holds:
	\begin{flalign*}
		\mdpval{\sigma}(b_i) &= \begin{cases}
			S_\mu & \text{ if } i \geq \mu \\
			S_1 & \text{ otherwise}
		\end{cases} \\
		\mdpval{\sigma}(F_i) &= \begin{cases}
			T_i & \text{ if } \bit_i^\oplus = 1 \\
			\varepsilon \cdot \BigO{1} + \frac12 \cdot (\pri{s} + S_\mu + \mdpval{\sigma}(b_{i-1})) & \text{ if } \bit_i^\mu = 0 \wedge \sigma(d_i) \not= F_i \\
			\varepsilon \cdot \BigO{1} + \pri{s} + S_\mu & \text{ otherwise}
		\end{cases}
	\end{flalign*}
	We conclude: $\mdpval{\sigma}(F_i) - \mdpval{\sigma}(b_{i-1}) = $
	\begin{displaymath}
	  \begin{cases}
	  	T_i - \mdpval{\sigma}(b_{i-1}) \geq T_i - S_\mu > 0 & \text{ if } S_\mu \\
	  	\varepsilon \cdot \BigO{1} + \frac12 \cdot (\pri{s} + S_\mu - \mdpval{\sigma}(b_{i-1})) \geq \varepsilon \cdot \BigO{1} + \frac12 \cdot \pri{s} > 0 & \text{ if } \bit_i^\mu = 0 \wedge \sigma(d_i) \not= F_i \\
	  	\varepsilon \cdot \BigO{1} + \pri{s} + S_\mu - \mdpval{\sigma}(b_{i-1}) \geq \varepsilon \cdot \BigO{1} + \pri{s} > 0 & \text{ otherwise}
	  \end{cases}
	\end{displaymath}

\item \emph{Phase 5}:
	The following holds:
	\begin{flalign*}
		\mdpval{\sigma}(g_i) &= \begin{cases}
			S_i & \text{ if } \bit_i^\oplus = 1 \\
			\varepsilon \cdot \BigO{1} + \pri{g_i} + \pri{s} + S_\mu & \text{ otherwise}
		\end{cases} \\
		\mdpval{\sigma}(a_i) &= \begin{cases}
			S_i & \text{ if } \bit_i^\oplus = 1 \wedge (i > \mu \vee \sigma(a_\mu) = g_\mu) \\
			\mdpval{\sigma}(a_{i+1}) & \text{ if } i = \mu \wedge \sigma(a_\mu) \not= g_\mu) \\
			S_\mu & \text{ if } i < \mu \wedge \sigma(a_i) \not= g_i \\
			\varepsilon \cdot \BigO{1} + \pri{g_i} + \pri{s} + S_\mu & \text{ if } i < \mu \wedge \sigma(a_i) = g_i
		\end{cases}
	\end{flalign*}
	By computing the difference we again see that the improving switches are as described.
\end{itemize}
\end{proof}

Now Lemma~\ref{lemma: policy proceeding} becomes applicable again and the proof of Theorem \ref{theorem: mdp}
follows.

\section{Linear Program Simplex Method Lower Bound}
\label{lp}

In this section we use the well known transformation from MDPs to linear programs to obtain
an exponential lower bound for the simplex method using Cunningham's rule.
\subsection{Linear Programs and the Simplex Method}
We briefly give a few basic definitions and state our notation. For more information the reader is referred
to any standard linear programming text, such as \cite{Chvatal}.
Given an $m$ by $n$ matrix $A$ with $m \le n$, an $n$-vector $c$ and $m$-vector $b$
we consider the primal linear program in the standard form
$$\begin{array}{ll}
\max & c^T x \\
\text{s.t.} & Ax = b \\
   &           x \ge 0
\end{array}$$
Let $B$ and $N$ be a partition of the indices $\{1,2, \ldots ,n\}$ such that $|B|=m$.
We denote by $A_B$ the submatrix of $A$ with columns indexed by $B$,
and $x_B$ the $m$-subvector of $x$ with indices in $B$. We say that
$B$ is a \emph{feasible basis} if $A_B$ is non-singular
and
\[
x_B = A_B^{-1} b \ge 0.
\]
A corresponding \emph{basic feasible solution(BFS)} $x$
is obtained if we extend $x_B$ to $x$ by setting $x_N =0$.

We call $c^T x$ the \emph{objective function}.
A \emph{pivot} from a feasible basis $B$ is defined by a pair of indices $i \in B$ and $j \in N$ for
which $B \setminus \{i\} \cup \{j\}$ is also a feasible basis. If the corresponding BFS are $x$ and $x'$,
then the pivot is \emph{improving} if $ c^T x \le c^T x'$. 
If the inequality is strict then we call the pivot step \emph{non-degenerate} otherwise it is called
\emph{degenerate}.
An \emph{optimal basis} is one
for which the corresponding BFS maximizes the objective function. 
A \emph{deterministic pivot rule} gives a unique pivot pair for 
every non-optimal feasible basis. 

The simplex method starts from a given feasible basis and applies improving pivots until
an optimal basis is obtained. The sequence of pivots depends on the specific pivot rule,
and care must be taken to ensure that it does not cycle if there are degenerate pivots.

The corresponding dual LP is written
$$\begin{array}{ll}
\min & b^T y \\
\text{s.t.} & A^T y \ge c. \\
\end{array}$$

\subsection{Markov Decision Processes as LPs}
\label{mdplp}

Optimal policies for MDPs that satisfy the unichain condition can be found by solving the following 
primal linear program (see, e.g., \cite{Puterman94}.)
$$\begin{array}{ll}
\max  & \sum_{(u,v)\in E_0} r(u,v)x(u,v) \\[1pt]
\text{s.t.} & \sum_{(u,v)\in E} x(u,v) - \sum_{(v,w)\in E_0,(w,u)\in E_R} p(w,u)x(v,w)  = 1 ,\, u\in V_0 \\[1pt]
            & x(u,v) \;\ge\; 0 \quad,\quad (u,v)\in E_0
\end{array}\leqno{(P)}$$
The variable $x(u,v)$, for $(u,v)\in E_0$, stands for the probability (frequency) of using the edge (action) $(u,v)$.
The constraints of the linear program are \emph{conservation constraints} that state that the probability of
entering a vertex~$u$ is equal to the probability of exiting~$u$.
It is not difficult to check that the BFS's of (P) correspond directly to policies of the MDP. 
For each policy $\sigma$ we can define a feasible setting of primal variables $x(u,v)$, for $(u,v)\in E_0$,
such that $x(u,v)>0$ only if $\sigma(u)=(u,v)$. Conversely, for every BFS $x(u,v)$ we can define a corresponding
policy $\sigma$. It is well known that the policy corresponding to an optimal BFS of (P) is an optimal policy
of the MDP. (See, e.g., \cite{Puterman94}.)
\begin{lemma}\label{lemma: primal bfs valuation}
Let $\sigma$ be a policy and $x(u,v)$ be a corresponding BFS. Then the following holds:
\[
\sum_{u \in E_0} \mdpval{\sigma}(u) = \sum_{(u,v)\in E_0} r(u,v)x(u,v)
\]
\end{lemma}

\noindent The dual linear program (for unichain MDPs) is:
$$\begin{array}{ll}
\min & \sum_{u \in V} y(u) \\
\text{s.t.} & y(u)-\sum_{(v,w)\in E_R} p(v,w)y(w) \;\ge\; r(u,v) \quad,\quad (u,v)\in E_0
\end{array}\leqno{(D)}$$
together with the condition that $y(u)$ sum up to~$0$ on the single irreducible recurrent class.

If $y^*$ is an optimal solution of (D), then $y^*(u)$, for every $u\in V_0$, is the value of~$u$
under an optimal policy. An optimal policy $\sigma^*$ can be obtained by letting
$\sigma^*(u)=(u,v)$, where $(u,v)\in E_0$ is an edge for which the inequality constraint in
(D) is tight, i.e., $y(u)-\sum_{w:(v,w)\in E_R} p(v,w)y(w) = r(u,v)$. Such a tight edge is guaranteed to exist.

\subsection{Policy Iteration and Simplex Method}

A policy iteration algorithm with policy $\sigma$ that perform a single switch at each iteration --
like Cunningham's rule --  corresponds to a variation of the simplex method
where the selection rule behaves like $\sigma$. Indeed $\sigma$
gives rise to a feasible solution $x(u,v)$ of the primal
linear program~(P) . We use $\sigma$ to define a Markov chain and let $x(u,v)$ be
the `steady-state' probability that the edge (action) $(u,v)$ is used. In
particular, if $\sigma(u)\ne v$, then $x(u,v)=0$.

We can also view the values corresponding to~$\sigma$ as settings of the variables
$y(u)$ of the dual linear program (D). By linear programming duality, if $y(u)$ is
feasible then $\sigma$ is an optimal policy. It is easy to check that an edge
$(u,v')\in E_0$ is an improving switch if and only if the dual constraint
corresponding to $(u,v')$ is violated. Furthermore, replacing the edge $(u,v)$ by
the edge $(u,v')$ corresponds to a pivoting step, with a non-negative reduced cost,
in which the column corresponding to $(u,v')$ enters the basis, while the column
corresponding to $(u,v)$ leaves the basis.

\subsection{Lower Bound Construction}

Let $n > 1$.
The variables of the LP correspond to the edges $E_0$ controlled by player~0, i.e.\ we have $10(n-1)$ variables
\begin{displaymath}
\{a_i^1, a_i^0, c_i^1, c_i^0, d_i^1, d_i^0 \mid 1 < i \leq n\} \cup \{b_i^1, b_i^0 \mid 1 < i < n\} \cup \{e_i^1, e_i^0 \mid 1 \leq i \leq n\}.
\end{displaymath}
The LP has $5(n - 1)$ constraints, corresponding to the nodes $V_0$ controlled by player~0, and labelled 
\begin{displaymath}
\{a_i, c_i, d_i \mid 1 < i \leq n\} \cup \{b_i \mid 1 < i < n\} \cup \{e_i \mid 1 \leq i \leq n\}.
\end{displaymath}
The linear program is defined as follows for each $n \ge 3$ (non-existent variables are assumed to be zero):
$$\begin{array}{ll}
LP_n: \\
\max  \sum_{i=1}^n \left(\left(a_i^1 + b_i^1 + c_i^1\right)  \left(\Omega(g_i) + \varepsilon  \Omega(h_i)\right) +
                    \varepsilon \left(d_i^1 + e_i^1\right)   \Omega(h_i) +
                    e_i^0 \Omega(s)\right) 
\end{array}
$$
$$\begin{array}{ll}
\text{subject to:}  \\
(a_2) \quad a_2^0 + a_2^1 = 1  + \varepsilon  ( b_2^0 + c_2^0 + d_2^0 + e_1^1 ) \\
(a_i) \quad a_i^0 + a_i^1 = 1 + a_{i-1}^0 + \varepsilon  (a_{i-1}^1 + b_{i-1}^1 + 
                  c_{i-1}^1 + d_{i-1}^1 + e_{i-1}^1)~~ & 3 \le i \le n\\
(b_i) \quad b_i^0 + b_i^1 = 1 + b_{i+1}^0 + d_{i+1}^0 & 2 \le i < n \\
(c_i) \quad c_i^0 + c_i^1 = 1 + c_{i+1}^0 & 2 \le i < n \\
(c_n) \quad c_n^0 + c_n^1 = 1 + \sum_{j=1}^n e_j^0 \\
(d_i) \quad d_i^0 + d_i^1 = 1 + \frac{1-\varepsilon}2 (a_i^1+b_i^1+c_i^1+d_i^1+e_i^1) & 2 \le i \le n \\
(e_1) \quad e_1^0 + e_1^1 = 1 + (1-\varepsilon) (b_2^0+c_2^0+d_2^0+e_1^1) \\
(e_i) \quad e_i^0 + e_i^1 = 1 +  \frac{1-\varepsilon}2 (a_i^1+b_i^1+c_i^1+d_i^1+e_i^1) & 2 \le i \le n \\
            ~~~~~~~~~~\text{ All variables non-negative}
\end{array}
$$
Note that the size of $LP_n$ is linear in $n$.
Depending on the context we let $LP_n$ denote both the linear program and the polytope defined
by its constraints.
For small values of $n$ explicit constructions of $LP_n$ and it dual are available online \cite{online}.

We now use the correspondence between the parity game $G_n$, the MDP $M_n$ and linear program $LP_n$
to get a lower bound for the simplex method using Cunningham's rule. 
The initial strategy for $G_n$ and initial policy for $M_n$, $\{ a_*^0, b_*^0, c_*^0, d_*^1, e_1^1, e_{*>1}^0 \}$,
defines a {\em starting basis} for $LP_n$. 
We construct a path on the polytope $LP_n$ from this starting basis
using the least recently considered rule with ordering (\ref{ordering}). 
Using this construction, Theorem \ref{theorem: mdp} implies that the path generated
will be in one-to-one correspondence with the sequence $P_n$ 
of improving switches generated in $M_n$ (and hence $G_n$).
We observe that the objective function strictly increases 
with each pivot due to Lemma~\ref{lemma: primal bfs valuation}
and Lemma~\ref{lemma: increasing valuation mdp}. Therefore we have the following result.
\begin{theorem}
\label{theorem: lp}
The pivot path $P_n$ for $LP_n$ from the starting basis to the optimum basis
followed by the least recently considered rule with ordering (\ref{ordering})
has length at least $2^n$. The objective function strictly increases with each pivot.
\end{theorem}

\section{Acyclic Unique Sink Orientations}
\label{auso}

Our final result concerns \emph{acyclic
unique sink orientations} (AUSOs), which are abstractions of various optimization
problems including linear programming, linear complementarity and binary
payoff games. In this section we extend our exponential lower bound to finding the sink
of an AUSO using the least considered rule.
For background information
on AUSOs, see \cite{SzWe01,Ga02,GaSc06}. 

\subsection{Definitions and previous results}
\label{ausodef}

AUSOs can be defined on arbitrary polytopes, but here we consider only hypercubes.
An AUSO on a $n$-dimensional hypercube is an orientation of its
edges that is acyclic and such that every face of the hypercube has a unique 
sink (vertex of outdegree 0).
The goal is to find the unique sink of the AUSO.
There is at present no known polynomial time algorithm for doing this, nor is it known to be NP-hard.

A natural class of algorithms to find the sink of an AUSO are \emph{path following algorithms}.
Such an algorithm would start at any vertex $v$
of the hypercube and repeatedly choose an outgoing edge according to some rule
until the unique sink is located. Each edge of the path corresponds to
flipping one bit of the current vertex. 

There is a very natural analogy between path following algorithms and
pivoting in linear programming. Pivot rules for
LPs therefore have natural analogues for AUSOs and
a full discussion of this is contained in \cite{AADMM12}.
In particular the least recently considered rule can be adapted to give
a path following algorithm to find the unique sink of an $n$-cube AUSO starting at any given vertex as follows.

We define, for each $i=1,2,...,n$, the variable $v_i$ to denote a flip of bit $i$ from 0 to 1, and
variable $v_{n+i}$ to denote a flip of bit $i$ from 1 to 0. We may now arrange the $2n$ variables
$v_1, v_2, ... , v_{2n}$ in any cyclic order. For any vertex of the AUSO that is not the sink
we find the first allowable flip in this cyclic order starting at the last chosen $v_i$. 

Suppose we are given a polytope and an objective function that is not constant on any edge of the polytope.
We can can then orient each edge in the direction of increasing objective function. The
corresponding directed graph on the skeleton of the polytope can be shown to be an AUSO.
The converse is not always true.
Indeed, we call an AUSO \emph{realizable} if there is a polytope and objective function that induces
a directed graph on its skeleton which is graph isomorphic to the AUSO.
Not all AUSOs are realizable.

The Klee-Minty examples\cite{Klee1972} are realizations of AUSOs on hypercubes, so exponential lower
bounds for most of the non-history based deterministic LP pivot rules immediately give similar bounds for AUSOs.
G{\"a}rtner\cite{Ga02} gives an $exp(2 n^{1/2})$ lower bound for random facet algorithms and
Matousek and
Szab{\'o}\cite{MatousekS04} give an $exp(const~ n^{1/3})$ lower bound for the
random edge rule on AUSOs.
In the next subsection we derive an exponential lower bound for finding the sink of a realizable AUSO using
a path following algorithm based on Cunningham's rule.

\subsection{Lower Bound Construction}

As we saw in Section \ref{parity} there is a direct relationship between binary
parity games and oriented 
hypercubes. Each vertex $v$ corresponds to a strategy ${\sigma}$ for player 0.
A partial orientation of the hypercube's edges is 
given by the notion of improving switches: the orientation goes from $\sigma$ to
$\sigma'$ iff there is a game edge $e$ 
belonging to player 0 s.t. $\sigma' = \sigma[e]$ and $e$ is
$\sigma$-improving. This is only a partial orientation since there are strategies
for the parity game that are not used in the lower bound construction and for which
the notion of improving is not well defined.
For each $n \ge 3$ we denote this partially oriented $n$-cube $H_n$.

Our goal is to embed $H_n$ into an AUSO $A_n$ whose edges orientations are consistent
with those already set in $H_n$. Note that this is not a trivial operation, as even for
$n=3$ it can be readily verified that there are partial acyclic orientations of the 3-cube
which are USOs on every complete face but do not embed into an AUSO.
We will achieve this embedding via the linear programming formulation of the last section,
achieving the stronger result that the AUSO is realizable.

\begin{lemma}
$LP_n$ is a realization an AUSO $A_n$ which is consistent with the edge orientations of $H_n$. 
\end{lemma}
\begin{proof}
We argued in Section \ref{mdplp} that each basic feasible solution of $LP_n$ corresponded 
directly to a policy of the corresponding MDP $M_n$, and hence to strategy of the
parity game $G_n$. Also each edge of $LP_n$ corresponded to a switch in $M_n$ and
hence to an edge owned by player 0 in $G_n$. It follows that the vertices
and edges of $LP_n$ and $H_n$ are in one-to-one 
correspondence. Since $H_n$ is an $n$-cube so is the skeleton of $LP_n$.
By applying symbolic dual perturbation if necessary to 
resolve ties in the objective function (see, e.g. \cite{Chvatal}) $A_n$ can be
oriented to give an AUSO, which we denote as $A_n$. By definition $A_n$ is realizable.

According to Theorem \ref{theorem: lp} the objective function
of $LP_n$ is non-constant on every edge and increases in the direction corresponding
to an improving switch in $M_n$. 
However improving switches in $M_n$ correspond
to improving edges in $G_n$. The edges of $H_n$ were directed in the same way 
as these improving edges. Since dual perturbation will not change the direction
of any edges for which the objective function is strictly increasing, the directed edges in $H_n$
maintain their directions in $A_n$.
The lemma follows.
\qed
\end{proof}

The lemma implies that not only is $H_n$ an AUSO but stronger properties, such as the 
Holt-Klee condition and the shelling property, also hold (see, e.g. \cite{AM09}).
For small values of $n$ the AUSOs $A_n$ are available online \cite{online}.

We may now apply the least recently considered rule to $A_n$ as described in Section \ref{ausodef}.
The starting basis of $LP_n$ $\{ a_*^0, b_*^0, c_*^0, d_*^1, e_*^0 \}$ defines a 
{\em starting vertex} of $A_n$.
We construct a path from this vertex
using the least recently considered rule with ordering (\ref{ordering}).
Using this construction, Theorem \ref{theorem: lp} implies that the path generated
will be in one-to-one correspondence with the path $P_n$ generated in $LP_n$.

\begin{theorem}
The directed path $P_n$ in $A_n$ from the starting vertex 
to the unique sink
followed by the least recently considered rule with ordering (\ref{ordering})
has length at least $2^n$.
\end{theorem}

\section{Conclusions}
\label{conclusion}

We have shown in this paper that Cunningham's least considered rule can lead
to exponential worst case behaviour in parity and other games, 
Markhov decision processes, linear programs and AUSOs.
This appears to be the first such result for a history based rule for AUSOs. 
However Cunningham's rule was in fact first proposed for the network simplex method.
The LPs presented in Section \ref{lp} are not networks, but are 
structurally remarkably close to them.
Our first open problem would be to extend our results to network LPs.

As remarked in the introduction, Zadeh's rule has recently been shown to 
have superpolynomial worst case behavior on linear programs. The parity games
behind these LPs were not binary, so it does not immediately follow that 
Zadeh's rule has similar behaviour on AUSOs. This is a second open problem.
More generally it is of interest determine whether all of the history based rules
mentioned in \cite{AADMM12} have exponential behaviour on AUSOs.

Finally it would be of interest to study the AUSOs presented in Section \ref{auso}
to see if they have a simple enough structure to be stated explicitly.
If so it may be possible to prove the existence of exponentially long paths in
a simpler manner than that done in Section \ref{parity}.




\bibliographystyle{spmpsci}
\bibliography{main}

\begin{thebibliography}{10}
\providecommand{\url}[1]{{#1}}
\providecommand{\urlprefix}{URL }
\expandafter\ifx\csname urlstyle\endcsname\relax
  \providecommand{\doi}[1]{DOI~\discretionary{}{}{}#1}\else
  \providecommand{\doi}{DOI~\discretionary{}{}{}\begingroup
  \urlstyle{rm}\Url}\fi

\bibitem{AADMM12}
Aoshima, Y., Avis, D., Deering, T., Matsumoto, Y., Moriyama, S.: On the
  existence of hamiltonian paths for history based pivot rules on acyclic
  unique sink orientations of hypercubes.
\newblock Discrete Applied Mathematics \textbf{160}(15), 2104--2115 (2012)

\bibitem{AM09}
Avis, D., Moriyama, S.: {On Combinatorial Properties of Linear Programming
  Digraphs}.
\newblock In: D.~Avis, D.~Bremner, A.~Deza (eds.) {Polyhedral Computation},
  {CRM Proceedings and Lecture Notes 48}, pp. 1--13. AMS (2009)

\bibitem{Bertsekas01}
Bertsekas, D.: Dynamic programming and optimal control, second edn.
\newblock Athena Scientific (2001)

\bibitem{Chvatal}
Chv{\'a}tal, V.: Linear Programmming.
\newblock W.H. Freeman (1983)

\bibitem{Cunningham79}
Cunningham, W.H.: Theoretical properties of the network simplex method.
\newblock In: Mathematics of Operations Research, pp. 196--208 (1979)

\bibitem{Derman72}
Derman, C.: Finite state {Markov} decision processes.
\newblock Academic Press (1972)

\bibitem{focs91*368}
Emerson, E., Jutla, C.: Tree automata, $\mu$-calculus and determinacy.
\newblock In: Proc.\ 32nd Symp.\ on Foundations of Computer Science, pp.
  368--377. IEEE, San Juan (1991)

\bibitem{FriedmannExpBound09}
Friedmann, O.: An exponential lower bound for the latest deterministic strategy
  iteration algorithms.
\newblock Logical Methods in Computer Science \textbf{7}(3) (2011)

\bibitem{Friedmann11}
Friedmann, O.: A subexponential lower bound for {Zadeh's} pivoting rule for
  solving linear programs and games.
\newblock In: IPCO, pp. 192--206 (2011)

\bibitem{Friedmann12}
Friedmann, O.: A subexponential lower bound for the least recently considered
  rule for solving linear programs and games.
\newblock In: GAMES'2012. Naples, Italy (2012)

\bibitem{online}
Friedmann, O.:  (2013).
\newblock \url{http://tcswiki.com/?title=Cunningham's_Rule}

\bibitem{Ga02}
G{\"a}rtner, B.: The random-facet simplex algorithm on combinatorial cubes.
\newblock Random Structures \& Algorithms \textbf{20}(3), 353--381 (2002)

\bibitem{GaSc06}
G{\"a}rtner, B., Schurr, I.: Linear programming and unique sink orientations.
\newblock In: Proceedings of the 17th Annual ACM-SIAM Symposium on Discrete
  Algorithms (SODA), pp. 749--757 (2006)

\bibitem{Howard60}
Howard, R.: Dynamic programming and {Markov} processes.
\newblock MIT Press (1960)

\bibitem{Kalai92}
Kalai, G.: A subexponential randomized simplex algorithm (extended abstract).
\newblock In: STOC, pp. 475--482 (1992)

\bibitem{Klee1972}
Klee, V., Minty, G.J.: {How Good is the Simplex Algorithm?}
\newblock In: O.~Shisha (ed.) {Inequalities III}, pp. 159--175. {Academic Press
  Inc.}, New York (1972)

\bibitem{MSW92}
Matousek, J., Sharir, M., Welzl, E.: A subexponential bound for linear
  programming.
\newblock In: Symposium on Computational Geometry, pp. 1--8 (1992)

\bibitem{MatousekS04}
Matousek, J., Szab{\'o}, T.: Random edge can be exponential on abstract cubes.
\newblock In: FOCS, pp. 92--100 (2004)

\bibitem{MC94}
Melekopoglou, M., Condon, A.: On the complexity of the policy improvement
  algorithm for markov decision processes.
\newblock INFORMS Journal on Computing \textbf{6}(2), 188--192 (1994)

\bibitem{PY12}
Post, I., Ye, Y.: The simplex method is strongly polynomial for deterministic
  markov decision processes.
\newblock CoRR \textbf{abs/1208.5083} (2012)

\bibitem{puri/phd}
Puri, A.: Theory of hybrid systems and discrete event systems.
\newblock Ph.D. thesis, University of California, Berkeley (1995).
\newblock
  \urlprefix\url{http://www.eecs.berkeley.edu/Pubs/TechRpts/1995/2950.html}

\bibitem{Puterman94}
Puterman, M.: Markov decision processes.
\newblock Wiley (1994)

\bibitem{SzWe01}
Szab{\'o}, T., Welzl, E.: Unique sink orientations of cubes.
\newblock In: Proceedings\ of the 42th FOCS, pp. 547--555 (2001)

\bibitem{conf/cav/VogeJ00}
V{\"o}ge, J., Jurdzinski, M.: A discrete strategy improvement algorithm for
  solving parity games.
\newblock In: Proc.\ 12th Int.\ Conf.\ on Computer Aided Verification,
  {CAV'00}, \emph{LNCS}, vol. 1855, pp. 202--215. Springer (2000)

\bibitem{Zadeh80}
Zadeh, N.: What is the worst case behavior of the simplex algorithm.
\newblock In: Polyhedral Computation, pp. 131--143. American Mathematical
  Society (2009,1980)

\bibitem{zwickpaterson/1996}
Zwick, U., Paterson, M.: The complexity of mean payoff games on graphs.
\newblock Theoretical Computer Science \textbf{158}(1-2), 343--359 (1996).
\newblock \doi{http://dx.doi.org/10.1016/0304-3975(95)00188-3}

\end{thebibliography}





\end{document}